%% file: generalized_gs.tex
\documentclass[11pt]{amsart}

\input{./preamble}

\title[Hartree-Fock ground state in magic angle graphene]{On the Hartree-Fock Ground State Manifold in Magic Angle Twisted Graphene Systems}

\author[K. D. Stubbs]{Kevin D. Stubbs}
\email{kstubbs@berkeley.edu}
\address{Department of Mathematics, University of California, Berkeley, CA 94720, USA}

\author[S. Becker]{Simon Becker}
\email{simon.becker@math.ethz.ch}
\address{ETH Zurich, 
Institute for Mathematical Research, 
Rämistrasse 101, 8092 Zurich, 
Switzerland}

\author[L. Lin]{Lin Lin}
\email{linlin@math.berkeley.edu}
\address{Department of Mathematics, University of California, Berkeley, CA 94720, USA; Applied Mathematics and Computational Research Division, Lawrence Berkeley National Laboratory, Berkeley, CA 94720, USA}

\allowdisplaybreaks

\newtheorem{result}{Result}
\newtheorem{lemma}{Lemma}[section]
\newtheorem{assumption}{Assumption}
\newtheorem{remark}{Remark}
\newtheorem{proposition}{Proposition}
\newtheorem{corollary}{Corollary}
\newtheorem*{conjecture*}{Conjecture}
\newtheorem{theorem}{Theorem}

\begin{document}

\maketitle

\begin{abstract}

Recent experiments have shown that magic angle twisted bilayer graphene (MATBG) can exhibit correlated insulator behavior at half-filling.
Seminal theoretical results towards understanding this phase in MATBG has shown that Hartree-Fock ground states (with a positive charge gap) can be exact many-body ground states of an idealized flat band interacting (FBI) Hamiltonian.  We prove that in the absence of spin and valley degrees of freedom, the only Hartree-Fock ground states of the FBI Hamiltonian for MATBG are two ferromagnetic Slater determinants. Incorporating spin and valley degrees of freedom, we provide a complete characterization of the Hartree-Fock ground state manifold, which is generated by a ${\rm U}(4) \times {\rm U}(4)$ hidden symmetry group acting on five elements. We also introduce new tools for ruling out translation symmetry breaking in the Hartree-Fock ground state manifold, which may be of independent interest.
\end{abstract}

\section{Introduction}
\label{sec:intro}
Over the past few years, ``magic angle'' twisted bilayer graphene (MATBG) has attracted immense attention in the condensed matter physics community due, in part, to recent experiments which demonstrate correlated insulating \cite{NatSaito,XieLianJackEtAl2019} and superconducting phases \cite{2018Nature}. The flat band interacting (FBI) Hamiltonian at the chiral limit, which is an idealized model of MATBG at the moir\'e scale\footnote{For MATBG, each moir\'e unit cell contains around $10^4$ carbon atoms. Throughout the paper,  details at the atomic scale are suppressed, and the focus is solely on the periodic structure at the moir\'e  scale.} \cite{TarnopolskyKruchkovVishwanath2019}, suggests that the ground state is both simple (at least in some regimes) and has a rich structure \cite{BultinckKhalafLiuEtAl2020,ChatterjeeBultinckZaletel2020,SoejimaParkerBultinckEtAl2020,XieMacDonald2020,WuSarma2020,DasLuHerzog-Arbeitman2021,BernevigSongRegnaultEtAl2021,LiuKhalafLee2021,SaitoGeRademaker2021,JiangLaiWatanabe2019,PotaszXieMacDonald2021,LiuKhalafLeeEtAl2021,FaulstichStubbsZhuEtAl2023,nuckolls2023quantum,wagner2022global,zhou2024kondo,xie2023phase,tseng2022anomalous}.
 On one hand, the Hartree-Fock states, which are the simplest quantum many-body states, can be  \emph{exact ground states} of the FBI Hamiltonian at half-filling\footnote{At other integer fillings, while Hartree-Fock states can be exact eigenstates of the FBI Hamiltonian, they may not be ground states.}. These quantum states are ferromagnetic, meaning the density matrix behaves identically at all $\vk$ points. On the other hand, the FBI Hamiltonian with spin and valley\footnote{Monolayer graphene features two nonequivalent Dirac points within the Brillouin zone. Near each Dirac point there is a small region, referred to as a valley. In MATBG, the wavefunction can be supported on two valleys simultaneously, referred to as intervalley coherent (IVC) states.} degrees of freedom exhibits a large ${\rm U}(4) \times {\rm U}(4)$ ``hidden symmetry'' \cite{BultinckKhalafLiuEtAl2020,LianSongRegnaultEtAl2021}. This implies that the ground state of the FBI Hamiltonian can, at most, be determined up to this symmetry. To the best of our knowledge, whether the ferromagnetic ground states (subject to a ${\rm U}(4) \times {\rm U}(4)$ symmetry) are the only Hartree-Fock ground states at half-filling remains an open question.

\subsection{Main Results}
\label{sec:main-results}
In this paper, we characterize the Hartree-Fock ground states of the flat band interacting (FBI) Hamiltonian at half-filling, which is derived from the Bistritzer-MacDonald model~\cite{BistritzerMacDonald2011,WatsonKongMacDonaldEtAl2022,CancesGarrigueGontier2023} at the chiral limit \cite{TarnopolskyKruchkovVishwanath2019,WatsonLuskin2021,BeckerEmbreeWittstenEtAl2021,BeckerEmbreeWittstenEtAl2022}.
 We refer the readers to~\cref{sec:flat-band-inter} for a brief review of the FBI Hamiltonian or \cite[Section 2]{BeckerLinStubbs2023} for a more complete review. 
Ref.~\cite{BeckerLinStubbs2023} proves that in a simplified model of MATBG without spin or valley degrees of freedom, among the Hartree-Fock at half-filling that are (1) uniformly filled, and (2) translation invariant, there are only two possible ground states called the ferromagnetic Slater determinant states.

This paper extends the findings of \cite{BeckerLinStubbs2023} in two  ways. First, we prove that in the absence of spin or valley degrees of freedom in MATBG, the two ferromagnetic Slater determinant states are the only Hartree-Fock ground states at half-filling. In particular, the ``uniform filling'' and ``translation invariant'' restrictions have been lifted. 
Second, when incorporating spin and valley degrees of freedom in MATBG, the manifold of Hartree-Fock ground states at half-filling is characterized by a \({\rm U}(4) \times {\rm U}(4)\) symmetry applied to five unique ferromagnetic Slater determinant states. This provides a complete characterization of the Hartree-Fock ground states at half-filling in MATBG.

Furthermore, the FBI Hamiltonian model can be generalized to other models of twisted bilayer graphene and twisted multilayer graphene. Specifically, the MATBG model above refers to the chiral limit twisted bilayer graphene Hamiltonian with two flat bands at zero energy (TBG-2) at a magic angle.  Ref.~\cite{BeckerLinStubbs2023} also analyzes  the Hartree-Fock ground state for a general class of FBI Hamiltonians under certain symmetry and non-degeneracy conditions. The non-degeneracy conditions in \cite{BeckerLinStubbs2023} are generally computationally intractable. For twisted bilayer graphene with four flat bands at a magic angle (TBG-4) and equal twist angle trilayer graphene with four flat bands (eTTG-4), the non-degeneracy conditions can be simplified and explicitly verified. As a byproduct of this study, the non-degeneracy condition is replaced by a new condition that is computationally more tractable. We also provide a sufficient condition to rule out translation symmetry breaking for these general systems. 

A key property of the FBI Hamiltonian for magic angle twisted graphene systems is that they can be \emph{frustration-free}.
In fact, the Hamiltonians for these systems can be expressed as a sum of positive semidefinite terms so that the ground state energy is bounded below by zero. 
Moreover, there exist ground states which minimize the energy of each individual term, also resulting in a value of zero. 
As a consequence of the form of the FBI Hamiltonian, any many-body ground state of a frustration-free FBI Hamiltonian is necessarily half-filled (\cref{prop:half-filling}).

Moving to Hartree-Fock theory, our first result is that the Hartree-Fock ground states of a frustration-free FBI Hamiltonian are always, in a certain sense, ``ferromagnetic'': 
\begin{result}[Informal Version of~\cref{prop:pseudo-ferromagnetism}]
  The Hartree-Fock ground states of a frustration-free FBI Hamiltonian are always ``ferromagnetic'', in the sense that knowing the 1-RDM at a single momentum determines the value of the 1-RDM at all other momenta.
\end{result}
This property can be used to infer that the Hartree-Fock ground states of FBI Hamiltonians are always uniformly filled.
We next give a condition to rule out translation breaking.
In particular, we show that the translation symmetry breaking of Hartree-Fock ground states of an FBI Hamiltonian is related to the solution of a family of coupled \textit{Sylvester} equations (see~\cref{sec:sylv-equat-transl} for more details).
This characterization of the Hartree-Fock ground states implies the following result:

\begin{result}[Informal Version of~\cref{thm:tbg-2,,thm:tbg-2-valley,,thm:tbg-2-valley-spin}]
  Let the frustration-free FBI Hamiltonian for TBG-2 be denoted by $\hat{H}_{\rm FBI}$.
  \begin{itemize}
  \item If $\hat{H}_{\rm FBI}$ is valleyless and spinless, then the two ferromagnetic Slater determinant states are the only Hartree-Fock ground states.
  \item If $\hat{H}_{\rm FBI}$ is valleyful and spinless, then all Hartree-Fock ground states are generated by the orbit of ${\rm U}(2) \times {\rm U}(2)$ on three elements.
  \item If $\hat{H}_{\rm FBI}$ is valleyful and spinful, then all Hartree-Fock ground states are generated by the orbit ${\rm U}(4) \times {\rm U}(4)$ on five elements.
  \end{itemize}
\end{result}
Finally, we provide an algorithm in \cref{prop:translation-breaking-alg} for determining a spanning set for all potential Hartree-Fock ground states.

\subsection{Discussion and open question}
In this paper, we find conditions which allow us to characterize all Hartree-Fock ground states of frustration-free FBI Hamiltonians.
We then apply these conditions to the special case of magic angle twisted bilayer graphene with two flat bands and extend this to include valley and spin degrees of freedom.
It would be interesting to apply the methodology here to other FBI Hamiltonians corresponding to twisted multilayer graphene systems.

Considering the simplicity of the FBI model and the low amount of electron correlation of the many-body ground state observed in numerical studies at half-filling \cite{SoejimaParkerBultinckEtAl2020,FaulstichStubbsZhuEtAl2023}, we propose the following conjecture:

\begin{conjecture*}[Many-Body Ground States of Frustration-Free FBI Hamiltonians]
Any many-body ground states of a frustration-free FBI Hamiltonian can be written as a linear combination of its Hartree-Fock ground states. In particular, for TBG-2 (with spin and valley), any many-body ground state can be written as 
\begin{equation}
\ket{\Psi}=\int_{{\rm U}(4) \times {\rm U}(4)} \sum_{i=1}^5 \alpha_{\mathfrak{g},i} \mathfrak{g} \ket{\Psi_i}\ud \mathfrak{g}, \quad \alpha_{\mathfrak{g},i}\in\CC.
\end{equation}
Here $\{\ket{\Psi_i}\}_{i=1}^5$ are the ferromagnetic Slater determinant states which generate the Hartree-Fock ground state manifold.
\end{conjecture*}

\subsection{Organization}
Our article is structured as follows:
\begin{itemize}
    \item In  \cref{sec:flat-band-inter} we define the flat-band interacting model for twisted graphene systems.
    \item In \cref{sec:hartree-fock-theory} we discuss the Hartree-Fock theory of the FBI Hamiltonian, study the Hartree-Fock energy, and find conditions for a state to be a Hartree-Fock ground state.
     \item In \cref{sec:pseudo-ferr-flat}, we show that the Hartree-Fock ground states of FBI Hamiltonians always satisfy a modified version of ferromagnetism as compared to the ferromagnetic Slater determinant ground states discussed in our previous work.
 \item In \cref{sec:sylv-equat-transl}, we show that translation breaking for a Hartree-Fock ground state can be characterized in terms of solutions to a family of coupled Sylvester equations~\cite[Chapter VII.2]{Bhatia1997}. We then show how this Sylvester equation allows us to understand the appearance of the ${\rm U}(2) \times {\rm U}(2)$ freedom in valleyful TBG and the ${\rm U}(4) \times {\rm U}(4)$ freedom in valleyful and spinful TBG.
\item Our article contains two appendices. \cref{sec:calc-hf-energy} contains the computations to characterize the Hartree-Fock energy and \cref{sec:trace-lemma-proof} contains the proof of a technical Lemma.

\end{itemize}

\subsection{Notation}
\label{sec:notation}
In this paper, we follow the same conventions as used in \cite{BeckerLinStubbs2023}.
Operators and matrices acting on the Fock space are represented using the hat notation, such as $\hat{f}^{\dag}, \hat{f}, \hat{H}_{\rm FBI}$.
Operators and matrices that operate in the single particle space (such as $L^2(\R^2;\CC^2\times \CC^2)$ for TBG) are indicated without the hat notation, such as operators $H$ and $D$.
With some slight abuse of the hat notation, vectors defined in real space are denoted without the hat, for example, $f(\vr)$.
Their corresponding Fourier transforms are indicated with the hat notation, as in $\hat{f}(\vq)$.
For any given matrix $A$, the operations of entrywise complex conjugation, transpose, and Hermitian conjugation are represented by $\overline{A}, A^{\top},$ and $A^{\dag}$, respectively. Additionally, the Frobenius norm of $A$ is denoted $\| A \|_{F}$. The identity is denoted by $I,$ where we occasionally write $I_{n\times n}$ to indicate the matrix size of the identity.

\subsection*{Acknowledgments}
This work was supported by the Simons Targeted Grants in Mathematics and Physical Sciences on Moir\'e Materials Magic (K.D.S., L.L.) and the SNF Grant PZ00P2 216019 (S.B.). L.L. is a Simons Investigator in Mathematics. We thank Dumitru C\u{a}lug\u{a}ru, Eslam Khalaf, Patrick Ledwidth,  Oskar Vafek and Michael Zaletel for helpful discussions.

\section{The Flat-Band Interacting Hamiltonian for Twisted Graphene}
\label{sec:flat-band-inter}
The Flat-Band Interacting (FBI) Hamiltonian for twisted bilayer graphene
is based on a periodic single particle Hamiltonian which has exactly flat bands.
We highlight the main features of this model and refer the reader to~\cite[Section 2]{BeckerLinStubbs2023} for more details.
Following the notation in \cite{BeckerLinStubbs2023}, we let $\Gamma$ and $\Omega$ denote the moir{\'e} lattice and real space unit cell in $\R^2$ respectively. Similar, let $\Gamma^{*}$ and $\Omega^{*}$ denote the moir{\'e} reciprocal lattice and Brillouin zone in $\R^2$ respectively.

Suppose that we are given a single particle Hamiltonian $H$ with flat-bands which is periodic with respect to $\Gamma$.
Given such a Hamiltonian, let $\{ \psi_{n\vk} : \vk \in \Omega^{*}, n \in \mc{N} \}$ be the Bloch eigenfunctions where $\mc{N}$ is the set indexing the flat bands.
For a system with $N$-layers, each $\psi_{n\vk}(\vr):=[\psi_{n\vk}(\vr; \sigma, j)] \in \CC^{2N},$ where $\sigma \in \{ A, B \}$ denotes the sublattice degree of freedom and $j \in \{ 1, \cdots, N \}$ denotes the layer degree of freedom.
 Furthermore, let $u_{n\vk}(\vr) = e^{- i \vk \cdot \vr} \psi_{n\vk}(\vr)$ denote the $\Gamma$-periodic Bloch functions normalized so that $\int_{\Omega} \|u_{n\vk}(\vr)\|^{2} \ud\vr = 1$, where $\| u \|:=\sqrt{\sum_{i=1}^n \vert u_i\vert^2}$ is the $2$-norm of a vector $u \in \mathbb C^n.$ Similarly, we denote by $\langle v,u\rangle:=\sum_{i=1}^n \overline{v_i}u_i$ the canonical inner-product.

The moir{\'e} Brillouin zone $\Omega$ is discretized using a finite grid denoted by $\mc{K}$
  \begin{equation}
    \label{eq:mcK}
    \mc{K} := \left\{ \frac{i}{n_{k_{x}}} \vg_1 + \frac{j}{n_{k_{y}}} \vg_2 : i \in \{0, 1, \cdots, n_{k_{x}} - 1\}, j \in \{0, 1, \cdots, n_{k_{y}} - 1 \} \right\} \subseteq \Omega^*.
  \end{equation}
  Here $\vg_1$, $\vg_2$ are a pair of generating vectors for the moir{\'e} reciprocal lattice $\Gamma^{*}$ and $(n_{k_{x}}, n_{k_{y}})$ are the number of points in each of the two lattice directions. 
  We also define $N_{\vk} := \# | \mc{K} | = n_{k_{x}} n_{k_{y}}$.

  While in this work $n_{k_{x}}, n_{k_{y}}$ are taken to be finite, we will be interested in the limit where both $n_{k_{x}}, n_{k_{y}}$ can become arbitrarily large towards the thermodynamic limit.
  We make the following technical assumption on the grid $\mc{K}$ in relation to the single particle Hamiltonian:
  \begin{assumption}
    \label{assume:grid}
    For each $\vk \in \mc{K}$, let $\Pi(\vk)$ be the orthogonal projector onto the flat-band eigenfunctions at $\vk$.
    We assume that the grid $\mc{K}$ has been chosen so that for all pairs of momenta $\vk, \vk' \in \mc{K}$ there exists a sequence of momenta $\{ \vk_{i} \}_{i=1}^L \subseteq \mc{K}$ so that $\vk_{1} = \vk$, $\vk_{L} = \vk'$ and $\| \Pi(\vk_{i+1}) - \Pi(\vk_{i}) \| < 1$ for all $i \in \{ 1, \ldots, L \}$.
  \end{assumption}
This assumption is satisfied for the chiral Bistritzer-MacDonald model \cite{TarnopolskyKruchkovVishwanath2019,BeckerEmbreeWittstenEtAl2022} so long as $n_{k_{x}}$ and $n_{k_{y}}$ are both chosen sufficiently large. This is because\ the flat bands are gapped from the remaining bands and the Hamiltonian depends real-analytically on $\vk.$

  Having fixed the flat-band eigenfunctions and the grid $\mc{K}$, we now define the FBI Hamiltonian.
  At each $\vk\in\mc{K}$, we first define the band creation and annihilation operators, $\hat{f}_{n\vk}^\dagger$ and $\hat{f}_{n\vk}$, which create or annihilate a particle in state $\psi_{n\vk}$ respectively.
  These many-body operators satisfy the canonical anti-commutation relation (CAR) as well as a periodicity condition:
  \begin{equation}
    \label{eq:car-definition}
    \begin{split}
      \{\hat{f}^{\dag}_{n\vk},\hat{f}_{n'\vk'}\}
      = \delta_{nn'} \delta_{\vk\vk'}, &\quad \{\hat{f}^{\dag}_{n\vk},\hat{f}^{\dag}_{n'\vk'}\}
                                         =\{\hat{f}_{n\vk},\hat{f}_{n'\vk'}\}=0 \\
      \hat{f}^{\dag}_{n(\vk+\vG)}=\hat{f}^{\dag}_{n\vk}, &\quad \hat{f}_{n(\vk+\vG)}=\hat{f}_{n\vk}, \quad \forall \vG\in\Gamma^*.
    \end{split}
  \end{equation}
  We also recall the definition of the number operator $\hat{N} := \sum_{\vk \in \mc{K}} \sum_{m \in \mc{N}} \hat{f}_{m\vk}^\dagger \hat{f}_{m\vk}$.
  
  Next, we define the Fourier coefficients of the periodic Bloch functions $u_{n\vk}$ introduced at the beginning of this section 
  \begin{equation}
    \hat{u}_{n\vk}(\vG; \sigma, j) :=  \int_{\Omega} e^{-i \vG \cdot \vr} u_{n\vk}(\vr; \sigma, j) \ud\vr = \int_{\Omega} e^{-i (\vk + \vG) \cdot \vr} \psi_{n\vk}(\vr; \sigma, j)  \ud\vr
  \end{equation}
  and define the \textit{form factor}
  \begin{equation}
    \label{eq:form-factor-def}
    \begin{split}
    [\Lambda_{\vk}(\vq')]_{mn} 
    & := \frac{1}{| \Omega |} \sum_{\vG' \in \Gamma^*} \braket{\hat{u}_{m\vk}(\vG'), \hat{u}_{n(\vk + \vq')}(\vG')} \\
    & = \frac{1}{| \Omega |} \sum_{\vG' \in \Gamma^*} \sum_{\sigma, j} \overline{\hat{u}_{m\vk}(\vG'; \sigma, j)} \hat{u}_{n(\vk + \vq')}(\vG'; \sigma, j)
    \end{split}
  \end{equation}
  where $\vk \in \mc{K}$, and $\vq' \in \mc{K} + \Gamma^{*}$.

  Finally, let $V$ be radially symmetric function with Fourier transform
  $\hat{V}(\vq) > 0$ for all $\vq\in \R^2$. 
  The FBI Hamiltonian for $H$ with Coulomb potential $V$ takes the form
  \begin{equation}
    \label{eq:h-fbi}
    \begin{split}
      \hat{H}_{\rm FBI} & = \frac{1}{N_{\vk} |\Omega|} \sum_{\vq' \in \mc{K} + \Gamma^*} \hat{V}(\vq') \widehat{\rho}(\vq') \widehat{\rho}(-\vq') \\
      \widehat{\rho}(\vq') & = \sum_{\vk \in \mc{K}} \sum_{m,n \in \mc{N}} [\Lambda_{\vk}(\vq')]_{mn} \left( \hat{f}_{m\vk}^\dagger \hat{f}_{n(\vk + \vq')} - \frac{1}{2} \delta_{mn} \delta_{\vq' \in \Gamma^{*}} \right).
    \end{split}
  \end{equation}
  The operator $\widehat{\rho}(\vq')$ satisfies  $\widehat{\rho}(-\vq') = \widehat{\rho}(\vq')^{\dagger}$ \cite[Lemma 5.1]{BeckerLinStubbs2023}, and so $\hat{H}_{\rm FBI}$ is positive semidefinite. Hence any state with zero energy must be a ground state.
  
  We note that from the form of the FBI Hamiltonian, we can immediately conclude that any zero energy ground state must be half-filled.
  \begin{lemma}
  \label{prop:half-filling}
    Any many-body ground state $\ket{\Phi}$ of a frustration-free FBI Hamiltonian is half-filled. That is, it satisfies
    \begin{equation}
        \hat{N} \ket{\Phi} = \frac{\# | \mc{N} |N_{\vk}}{2} \ket{\Phi} 
    \end{equation}
    where $\hat{N}$ is the number operator $\hat{N} = \sum_{\vk \in \mc{K}} \sum_{m \in \mc{N}} \hat{f}_{m\vk}^\dagger \hat{f}_{m\vk}$.
  \end{lemma}
  \begin{proof}
     Let $\hat{H}_{\rm FBI}$ be a frustration-free Hamiltonian, any many-body ground state $\ket{\Phi}$ must satisfy $\braket{\Phi | \widehat{\rho}(\vq') \widehat{\rho}(-\vq') | \Phi} = 0$.
      From the definition of $\Lambda_{\vk}(\vq')$, it can be verified that for all $\vk \in \Omega^*$, $\Lambda_{\vk}(\vzero) = I$.
      In particular,
      \begin{equation}
           \widehat{\rho}(\vq' = \vzero) = \sum_{\vk \in \mc{K}} \sum_{m \in \mc{N}} \left( \hat{f}_{m\vk}^\dagger \hat{f}_{m\vk} - \frac{1}{2} \right) = \hat{N} - \frac{\# | \mc{N} |N_{\vk}}{2}.
      \end{equation}
       To be a ground state, $\ket{\Phi}$ must satisfy
       \begin{equation}
            \braket{ \Phi | \widehat{\rho}(\vzero) \widehat{\rho}(\vzero) | \Phi} = \left\| \left( \hat{N} - \frac{\# | \mc{N} |N_{\vk}}{2} \right) \ket{\Phi}  \right\|^2 = 0.
      \end{equation}
      Hence, $\hat{N} \ket{\Phi} = \frac{\# | \mc{N} |N_{\vk}}{2} \ket{\Phi} $ as was claimed.
  \end{proof}

  \section{Hartree-Fock Theory for Flat-Band Interacting Hamiltonians}
  \label{sec:hartree-fock-theory}
  We begin this section by first reviewing the Hartree-Fock theory at half-filling without translation symmetry in \cref{sec:revi-non-transl}.
  We then derive a simple formula for the Hartree-Fock energy of the FBI Hamiltonian in \cref{sec:hartree-fock-energy}.
  As a consequence of this formula, we will obtain simple necessary and sufficient conditions for a Hartree-Fock state to be a ground state of a frustration-free FBI Hamiltonian.

  \subsection{A Review of Hartree-Fock Theory without Translation Symmetry}
  \label{sec:revi-non-transl}
  
  For a system with an even number of bands, recall that a general Slater determinant at half-filling takes the form
  \begin{equation}
    \label{eq:hf-sd}
    \ket{\Psi_{S}} = \prod_{i=1}^{M N_{\vk}} \hat b_{i}^{\dag}\ket{\mathrm{vac}},
  \end{equation}
  where $\ket{\mathrm{vac}}$ is the vacuum state, $M$ is half the number of available bands, and 
  \begin{equation}
    \label{eq:hf-orbitals}
    \hat b_{i}^{\dag}=\sum_{n \in \mc{N}} \sum_{\vk \in \mc{K}} \hat{f}_{n\vk}^{\dag} \Xi_{ni}(\vk), \qquad \sum_{n \in \mc{N}} \sum_{\vk \in \mc{K}} \overline{\Xi_{ni}(\vk)} \Xi_{nj}(\vk) = \delta_{ij}
  \end{equation}
  defines the creation operator for the Hartree-Fock orbitals.

  The Hartree-Fock equations can be expressed in terms of the one-body reduced density matrix (1-RDM). The 1-RDM associated with a given Slater determinant $\ket{\Psi_{S}}$ can be written as
  \begin{equation}
    [P(\vk',\vk)]_{nm}=\braket{\Psi|\hat{f}_{m\vk}^{\dagger}\hat{f}_{n\vk'}|\Psi}=
    \sum_{i=1}^{M N_{\vk}} \Xi_{ni}(\vk')\overline{\Xi_{mi}(\vk)}.
  \end{equation}
  Note that due to the periodicity of the creation and annihilation operators in \cref{eq:car-definition}, the 1-RDM satisfies the property $P(\vk + \vG, \vk' + \vG') = P(\vk, \vk')$ for all $\vG, \vG' \in \Gamma^{*}$. 

  The 1-RDM $[P(\vk, \vk')]_{mn}$ may be represented as a $(2 M N_{\vk}) \times (2 M N_{\vk})$ matrix as follows
  \begin{equation}
    \label{eq:1rdm-matrix}
    P
    =
    \begin{bmatrix}
      P(\vk_{1}, \vk_{1}) & P(\vk_{1}, \vk_{2}) & \cdots & P(\vk_{1}, \vk_{N_{\vk}}) \\
      P(\vk_{2}, \vk_{1}) & P(\vk_{2}, \vk_{2}) & \cdots & P(\vk_{2}, \vk_{N_{\vk}}) \\
      \vdots &  \vdots & \ddots & \vdots \\
      P(\vk_{N_{\vk}}, \vk_{1}) & \cdots & \cdots & P(\vk_{N_{\vk}}, \vk_{N_{\vk}}) \\
    \end{bmatrix}.
  \end{equation}
  Due to the orthogonality relation on $\Xi(\vk)$ given in~\cref{eq:hf-orbitals}, one may verify that $P$ satisfies $P^{2} = P$, $P^{\dagger} = P$, and $\Tr{(P)} = M N_{\vk}$ (i.e. $P$ is an orthogonal projection onto a $(M N_{\vk})$-dimensional vector space).
  The 1-RDM is \textit{translation-invariant} if it takes the form
    \begin{equation}
      \begin{bmatrix}
        P(\vk_{1}, \vk_{1}) & \vzero & \cdots & \vzero \\
        \vzero & P(\vk_{2}, \vk_{2}) & \cdots & \vzero \\
        \vdots &  \vdots & \ddots & \vdots \\
        \vzero & \cdots & \cdots & P(\vk_{N_{\vk}}, \vk_{N_{\vk}})
      \end{bmatrix}.
    \end{equation}

  \subsection{Hartree-Fock Energy of Flat-Band Interacting Hamiltonians}
  \label{sec:hartree-fock-energy}
  The Hartree-Fock energy for FBI Hamiltonians can also be written in terms of the shifted 1-RDM, $Q(\vk, \vk')$, defined as follows:
  \begin{equation}
    \label{eq:q-def}
    [Q(\vk,\vk')]_{mn} := [P(\vk, \vk')]_{mn} - \frac{1}{2} \delta_{mn} \delta_{\vk-\vk' \in \Gamma^{*}}.
  \end{equation}
  Since $P^{\dagger} = P$ and $P^{2} = P$, one easily checks that $Q^{\dagger} = Q$ and $Q^{2} = \frac{1}{4} I$.
  Furthermore, $Q(\vk, \vk')$ inherits the property $Q(\vk + \vG, \vk' + \vG') = Q(\vk, \vk')$ for all $\vG, \vG' \in \Gamma^{*}$ from the 1-RDM $P$.

The Hartree-Fock energy can be written in terms of $Q(\vk, \vk')$ as follows (\cref{sec:calc-hf-energy}) :
  \begin{equation}
    \label{eq:hf-energy2}
    \begin{split}
      \langle & \Psi | \hat{H}_{\rm FBI} | \Psi \rangle =  \frac{1}{N_{\vk} |\Omega|} \sum_{\vq' \in \Gamma^{*} + \mc{K}} \hat{V}(\vq') \left| \sum_{\vk \in \mc{K}} \Tr{\Big( \Lambda_{\vk}(\vq') Q(\vk + \vq', \vk) \Big)} \right|^{2} \\
              & + \frac{1}{4 N_{\vk} |\Omega|} \sum_{\vq' \in \Gamma^{*} + \mc{K}} \hat{V}(\vq') \sum_{\vk \in \mc{K}} \| \Lambda_{\vk}(\vq') \|_{F}^{2} \\
              & - \frac{1}{N_{\vk} |\Omega|} \sum_{\vq' \in \Gamma^{*} + \mc{K}} \hat{V}(\vq') \sum_{\vk, \vk' \in \mc{K}} \Tr{\Big( \Lambda_{\vk}(\vq') Q(\vk + \vq', \vk') \Lambda_{\vk'-\vq'}(\vq')^{\dagger} Q(\vk' -\vq', \vk)\Big)}. 
    \end{split}
  \end{equation}
  We can further simplify this expression by expanding the families of $(2M) \times (2M)$ matrices $\{ \Lambda_{\vk'}(\vq') : \vk' \in \mc{K} \}$ and $\{ Q(\vk, \vk') : \vk, \vk' \in \mc{K} \}$ into larger $(2MN_{\vk}) \times (2MN_{\vk})$ matrices similar to~\cref{eq:1rdm-matrix}.
  We start by defining notation for the standard basis on the space of 1-RDMs
  \begin{equation}
    \{ \ket{n, \vk} :  n \in \mc{N}, \vk \in \mc{K}\}
  \end{equation}
  where we take the convention that if $\vk + \vq' \in \mc{K} + \Gamma^{*}$ then
  \begin{equation}
    \ket{n, \vk + \vq'} := \ket{n, \widetilde{\vk + \vq'}}
  \end{equation}
  where $\widetilde{\vk + \vq'} \in \mc{K}$ is the unique representative of $\vk + \vq'$ in $\mc{K}$ modulo $\Gamma^{*}$ lattice translations.
  We additionally define the momentum shift matrix:
  \begin{equation}
    \pi_{\vq'} := \sum_{\vk \in \mc{K}} \sum_{m \in \mc{N}} \ket{m, \vk} \bra{m, \vk + \vq'}.
  \end{equation}
  This shifts momentum $\vk$ to momentum $\vk + \vq'$ modulo $\Gamma^{*}$ in the basis $\ket{m, \vk}$.

  Using this basis, for any $\vq'$ we define the matrix
  \begin{equation}
    \Lambda(\vq') := \sum_{\vk \in \mc{K}} \sum_{m,n \in \mc{N}} [\Lambda_{\vk}(\vq')]_{mn} \ket{m, \vk} \bra{n, \vk}
  \end{equation}
  which can be expressed as a direct sum as follows
  \begin{equation}
    \Lambda(\vq') =
    \bigoplus_{n=1}^{N_{\vk}} \Lambda_{\vk_{n}}(\vq')
    =
    \begin{bmatrix}
      \Lambda_{\vk_{1}}(\vq') & & & & \\
                              & \Lambda_{\vk_{2}}(\vq') & & & \\
                              && \Lambda_{\vk_{3}}(\vq') & & \\
                              &&& \ddots & \\
                              &&&& \Lambda_{\vk_{N_{\vk}}}(\vq')
                                   
    \end{bmatrix}.
  \end{equation}
  Furthermore, the matrix $Q$ can be written in this basis as
  \begin{equation}
    Q := \sum_{\vk,\vk' \in \mc{K}} \sum_{m,n \in \mc{N}} [Q(\vk, \vk')]_{mn} \ket{m, \vk} \bra{n, \vk'}.
  \end{equation}
  Note that in this basis $Q = P - \frac{1}{2} I$. Now observe that 
  \begin{equation}
    \begin{split}
      \pi_{\vq'} Q
      & = \sum_{\vk,\vk' \in \mc{K}} \sum_{m,n \in \mc{N}} [Q(\vk, \vk')]_{mn} \ket{m, \vk - \vq'} \bra{n, \vk'} \\
      & = \sum_{\vk,\vk' \in \mc{K}} \sum_{m,n \in \mc{N}} [Q(\vk + \vq', \vk')]_{mn} \ket{m, \vk} \bra{n, \vk'},
    \end{split}
  \end{equation}
  where the shift in $\ket{m, \vk - \vq'}$ in the first line is due to the multiplication on the right by $\pi_{\vq'}$.
  Performing the change of variables $\vk \mapsto \vk + \vq'$ is justified, since $\vq' \in \mc{K} + \Gamma^{*}$ and $Q(\vk, \vk')$ is invariant under shifts by $\Gamma^{*}$ lattice vectors.

  Similarly,
  \begin{equation}
    \begin{split}
      \pi_{\vq'}^{\dagger} \Lambda(\vq')^{\dagger}
      & = \sum_{\vk' \in \mc{K}} \sum_{m, n \in \mc{N}} [\Lambda_{\vk'}(\vq')^{\dagger}]_{mn} \ket{m, \vk' + \vq'} \bra{n, \vk'} \\
      & = \sum_{\vk' \in \mc{K}} \sum_{m, n \in \mc{N}} [\Lambda_{\vk' - \vq'}(\vq')^{\dagger}]_{mn} \ket{m, \vk'} \bra{n, \vk' -\vq'} \\
    \end{split}
  \end{equation}
  where we have used the fact that $\Lambda_{\vk}(\vq') = \Lambda_{\vk + \vG}(\vq')$ for any $\vG \in \Gamma^{*}$ (see \cref{lem:form-factor-identities}) to perform the change of variables $\vk' \mapsto \vk' - \vq'$.

  Therefore, the energy of a state $\ket{\Psi}$ can be written
  \begin{equation}
    \label{eq:hf-energy}
    \begin{split}
      \braket{\Psi | \hat{H}_{\rm FBI} | \Psi} =
      &  \frac{1}{N_{\vk} |\Omega|} \sum_{\vq' \in \Gamma^{*} + \mc{K}} \hat{V}(\vq') \left| \Tr{\Big( \Lambda(\vq') \pi_{\vq'} Q \Big)} \right|^{2} \\
      & + \frac{1}{4 N_{\vk} |\Omega|} \sum_{\vq' \in \Gamma^{*} + \mc{K}} \hat{V}(\vq') \| \Lambda(\vq') \|_{F}^{2} \\
      & -\frac{1}{N_{\vk} |\Omega|} \sum_{\vq' \in \Gamma^{*} + \mc{K}} \hat{V}(\vq') \Tr{\bigg( \Lambda(\vq') \pi_{\vq'} Q^{\dagger} \pi_{\vq'}^{\dagger} \Lambda(\vq')^{\dagger} Q \bigg)}.
    \end{split}
  \end{equation}
  We now recall the following elementary identity, whose proof is in \cref{sec:trace-lemma-proof}.
  \begin{lemma}
    \label{lem:trace-lemma}
    Let $A, B$ be arbitrary matrices then we have the following equality
    \begin{equation}
      \Re{\left(\Tr{\left( A B A^{\dagger} B^{\dagger} \right)}\right)} = \frac{1}{2} \Tr{\Big(AA^{\dagger} B^{\dagger} B + A^{\dagger} A B B^{\dagger}\Big)} - \frac{1}{2} \| [ A, B ] \|_{F}^{2}.
    \end{equation}
  \end{lemma}
  Since the energy of a state is purely real, we can set $A = \Lambda(\vq') \pi_{\vq'}$ and $B = Q^{\dag}$ and use~\cref{lem:trace-lemma} to rewrite~\cref{eq:hf-energy} as
  \begin{equation}
    \begin{split}
      \langle & \Psi | \hat{H}_{\rm FBI} | \Psi \rangle \\
              & = \frac{1}{N_{\vk} |\Omega|} \sum_{\vq' \in \mc{K} + \Gamma^{*}} \hat{V}(\vq') \left| \Tr{\Big( \Lambda(\vq') \pi_{\vq'} Q \Big)} \right|^{2} \\
              & + \frac{1}{4 N_{\vk} |\Omega|} \sum_{\vq' \in \mc{K} + \Gamma^{*}} \hat{V}(\vq') \| \Lambda(\vq') \|_{F}^{2} \\
              & -\frac{1}{2N_{\vk} |\Omega|} \sum_{\vq' \in \mc{K} + \Gamma^{*}} \hat{V}(\vq') \Tr{\bigg( \pi_{\vq'}^{\dagger} \Lambda(\vq')^{\dagger} \Lambda(\vq') \pi_{\vq'} Q Q^{\dagger} + \Lambda(\vq') \Lambda(\vq')^{\dagger} Q^{\dagger} Q \bigg)} \\
              & +\frac{1}{2 N_{\vk} |\Omega|} \sum_{\vq' \in \mc{K} + \Gamma^{*}} \hat{V}(\vq') \| [ Q, \Lambda(\vq') \pi_{\vq'} ] \|_{F}^{2}.
    \end{split}
  \end{equation}
  But for Hartree-Fock states $QQ^{\dagger} = Q^{\dagger} Q = \frac{1}{4} I $, the middle two terms cancel,  leaving us with
  \begin{equation}
    \label{eq:hf-energy-commutator}
    \begin{split}
      \langle & \Psi | \hat{H}_{\rm FBI} | \Psi \rangle \\
              & =  \frac{1}{N_{\vk} |\Omega|} \sum_{\vq' \in \mc{K} + \Gamma^{*}} \hat{V}(\vq') \left( \left| \Tr{\Big( \Lambda(\vq') \pi_{\vq'} Q \Big)} \right|^{2} + \frac{1}{2} \| [ Q, \Lambda(\vq') \pi_{\vq'} ] \|_{F}^{2} \right) \\
              & =  \frac{1}{N_{\vk} |\Omega|} \sum_{\vq' \in \mc{K} + \Gamma^{*}} \hat{V}(\vq') \left( \left| \Tr{\Big( \Lambda(\vq') \pi_{\vq'} \Big(P - \frac{1}{2} I\Big) \Big)} \right|^{2} + \frac{1}{2} \| [ P, \Lambda(\vq') \pi_{\vq'} ] \|_{F}^{2} \right)
    \end{split}
  \end{equation}
  where in the last line we have used $Q = P - \frac{1}{2} I$ and that the identity commutes with all matrices.
  Hence a Hartree-Fock state $\ket{\Psi}$ is a ground state of $\hat{H}_{\rm FBI}$ if and only if the following two conditions hold for all $\vq' \in \mc{K} + \Gamma^{*}$
  \begin{align}
    \Tr{\Big( \Lambda(\vq') \pi_{\vq'} \Big(P - \frac{1}{2} I\Big) \Big)} & = 0, \label{eq:hf-trace-condition} \\[1ex]
    [ P, \Lambda(\vq') \pi_{\vq'} ] & = 0. \label{eq:hf-comm-condition}
  \end{align}
  While the trace condition~\cref{eq:hf-trace-condition}~can be used to eliminate certain states as ground states, the commutator condition~\cref{eq:hf-comm-condition} is a powerful tool for characterizing possible Hartree-Fock ground states.

  \section{Generalized Ferromagnetism in Flat-Band Interacting Hamiltonians}
  \label{sec:pseudo-ferr-flat}
  In a ferromagnet there are two ground states with equal energy: $\ket{\uparrow\uparrow\uparrow\uparrow\uparrow\cdots}$ and $\ket{\downarrow\downarrow\downarrow\downarrow\downarrow\cdots}$.
  For such ground states, by knowing the value of the spin at a single site, we can determine the value at all other sites.
  Similar to a ferromagnet, if $\ket{\Psi}$ is a Hartree-Fock ground state of~\cref{eq:h-fbi}, then by knowing the value of the 1-RDM at a single momenta there is an explicit method to compute 1-RDM at other momenta.
  We term this property ``generalized ferromagnetism'' and formally have the following proposition:
  \begin{proposition}[Generalized Ferromagnetism]
    \label{prop:pseudo-ferromagnetism}
    Suppose that $P$ is the 1-RDM of a Hartree-Fock ground state of a frustration-free FBI Hamiltonian~\cref{eq:h-fbi} and the grid $\mc{K}$ satisfies~\cref{assume:grid}.
    For any $\vk_{1}, \vk_{2} \in \mc{K}$ there exists a family of invertible matrices $\{ B_{\vk}(\vq') : \vk \in \mc{K}, \vq' \in \R^{2} \}$ so that
    \begin{equation}
      \label{eq:symmetry}
      P(\vk_{1} + \vq', \vk_{2} + \vq') = B_{\vk_{1}}(\vq')^{-1} P(\vk_{1}, \vk_{2}) B_{\vk_{2}}(\vq').
    \end{equation}
    Hence, if the value of $P(\vk_{1}, \vk_{2})$ is known, then $P(\vk_{1} + \vq', \vk_{2} + \vq')$ is determined for all $\vq'$.
  \end{proposition}
  An immediate implication of~\cref{prop:pseudo-ferromagnetism} is that the Hartree-Fock ground states must be uniformly filled.
  \begin{corollary}
    \label{coro:uniform-filling}
    If $P$ is as in~\cref{prop:pseudo-ferromagnetism} then for all $\vk, \vk' \in \mc{K}$, $\Tr{(P(\vk, \vk))} = \Tr{(P(\vk',\vk'))}$.
  \end{corollary}
  \begin{proof}[Proof of~\cref{coro:uniform-filling}]
    Taking $\vk_{1} = \vk_{2}$ in the statement of~\cref{prop:pseudo-ferromagnetism} we see that
    \begin{equation}
      \begin{split}
        \Tr{(P(\vk_{1} + \vq', \vk_{1} + \vq'))} = \Tr{(B_{\vk_{1}}(\vq')^{-1} P(\vk_{1}, \vk_{1}) B_{\vk_{1}}(\vq'))} = \Tr{(P(\vk_{1}, \vk_{1}))}.
      \end{split}
    \end{equation}
    Since the choice of $\vq'$ was arbitrary, this proves the result.
  \end{proof}
  The proof of~\cref{prop:pseudo-ferromagnetism} relies on the following  lemma proven in our previous work:
  \begin{lemma}[Lemma 7.1 \cite{BeckerLinStubbs2023}]
    \label{lem:full-rank}
    Let $\vk, \vk' \in \mc{K}$ and let $\Pi(\vk)$ denote the flat-band projection at $\vk$.
    If $\| \Pi(\vk) - \Pi(\vk') \| < 1$ then there exists a $\vG \in \Gamma^*$ so that $\Lambda_{\vk}((\vk' - \vk) + \vG)$ has full rank.
  \end{lemma}
  We now turn to prove \cref{prop:pseudo-ferromagnetism}.
  \begin{proof}[Proof of~\cref{prop:pseudo-ferromagnetism}]
    Let us first fix some $\vk, \vk' \in \mc{K}$ and some $\vq' \in \mc{K} + \Gamma^{*}$.
    We can then decompose $\vq' = \vq + \vG$ where $\vq \in \mc{K}$ and $\vG \in \Gamma^{*}$.
    For this choice of $\vq + \vG$, \cref{eq:hf-comm-condition} implies that for all $\vk, \vk' \in \mc{K}$ 
    \begin{equation}
      \label{eq:hf-comm-equation}
      P(\vk, \vk') \Lambda_{\vk'}(\vq + \vG) - \Lambda_{\vk}(\vq + \vG) P(\vk + \vq, \vk' + \vq) = 0.
    \end{equation}
    If $\Lambda_{\vk}(\vq + \vG)$ is invertible then we have that
    \begin{equation}
      \label{eq:fbi-ferromagnetic}
      P(\vk + \vq, \vk' + \vq) = \Lambda_{\vk}(\vq + \vG)^{-1} P(\vk, \vk') \Lambda_{\vk'}(\vq + \vG).
    \end{equation}
    While we cannot assume that $\Lambda_{\vk}(\vq + \vG)$ is invertible in general, if we choose $\vq$ sufficiently small, then by~\cref{lem:full-rank} there exists a $\vG$ so that $\Lambda_{\vk}(\vq + \vG)$ must be invertible.

    By~\cref{assume:grid}, we can construct a sequence $\{ \vk_{i} \}_{i=1}^{L}$ which connects $\vk$ to $\vk + \vq$ so that $\| \Pi(\vk_{i+1}) - \Pi(\vk_{i}) \| < 1$ for $i = 1, \cdots, L - 1$.
    Along this path, there always exists a $\vG_{i}$ so that the form factor matrix $\Lambda_{\vk_{i}}(\vk_{i+1} - \vk_{i} + \vG_{i})$ is invertible.
    Therefore, we can prove the proposition by recursively applying~\cref{eq:fbi-ferromagnetic}.

    To make the induction step more clear, we define the momentum difference $\Delta := \vk' - \vk$ and the sequence of differences $\delta_{i} = \vk_{i+1} - \vk_{i}$.
    By~\cref{lem:full-rank}, there exists a $\vG_{1} \in \Gamma^{*}$ so that $\Lambda_{\vk_{1}}(\delta_{1} + \vG_{1})$ is invertible.
    Therefore, by~\cref{eq:fbi-ferromagnetic} 
    \begin{equation}
      \begin{split}
        P(\vk_{2}, \vk_{2} + \Delta) & = P(\vk_{1} + \delta_{1}, \vk_{1} + \Delta + \delta_{1}) \\
                                     & = \Lambda_{\vk_{1}}(\delta_{1} + \vG_{1})^{-1} P(\vk_{1}, \vk_{1} + \Delta) \Lambda_{\vk_{1} + \Delta}(\delta_{1} + \vG_{1}).
      \end{split}
    \end{equation}
    We can then apply the same argument to the next momentum in the sequence $\vk_{3}$
    \begin{equation}
      \begin{split}
        P(\vk_{3}, \vk_{3} + \Delta)
        & = P(\vk_{2} + \delta_{2}, \vk_{2} + \Delta + \delta_{2}) \\
        & = \Lambda_{\vk_{2}}(\delta_{2} + \vG_{2})^{-1} P(\vk_{2}, \vk_{2} + \Delta) \Lambda_{\vk_{2} + \Delta}(\delta_{2} + \vG_{2}).
      \end{split}
    \end{equation}
    Thus, by repeatedly applying~\cref{eq:fbi-ferromagnetic} along the path $\{ \vk_{i} \}_{i=1}^{L}$ we conclude
    \begin{equation}
      \label{eq:ferromagnetic-calc}
      P(\vk_{L}, \vk_{L} + \Delta) = \left(\prod_{i=1}^{L-1} \Lambda_{\vk_{i}}(\delta_{i} + \vG_{i}) \right)^{-1} P(\vk_{1}, \vk_{1} + \Delta) \left( \prod_{i=1}^{L-1} \Lambda_{\vk_{i} + \Delta}(\delta_{i} + \vG_{i}) \right).
    \end{equation}
    Therefore, since $\vk_{1} = \vk$, $\Delta = \vk' - \vk$, and $\vk_{L} = \vk + \vq$ this proves the proposition.
  \end{proof}

  \section{The Sylvester Equation and Translation Breaking}
  \label{sec:sylv-equat-transl}
  While we considered~\cref{eq:hf-comm-equation} with $\vq \neq \vzero$ in~\cref{sec:pseudo-ferr-flat}, let us now consider the case when $\vq = \vzero$.
  By choosing $\vq = \vzero$, from~\cref{eq:hf-comm-equation} we see that to be a ground state $P(\vk, \vk')$ must satisfy
  \begin{equation}
    \label{eq:commutator-q0}
    P(\vk, \vk') \Lambda_{\vk'}(\vG) - \Lambda_{\vk}(\vG) P(\vk, \vk') = 0 \quad \forall \vk, \vk' \in \mc{K},~ \forall \vG \in \Gamma^{*}.
  \end{equation}
  The equation $X A - B X = C$ for some fixed $A, B, C$ is known as the Sylvester equation. A simple condition ensures that this equation has a unique solution (see \cite[Theorem VII.2.1]{Bhatia1997} with $C=0$).
  \begin{lemma}
    \label{lem:sylvester}
    Consider the Sylvester-type equation with $A, B \in \CC^{N \times N}$,
    \begin{equation}
      \label{eq:sylvester}
      X A - B X = 0 \text{ for }X \in \mathbb C^{N \times N}.
    \end{equation}
    The value $X = 0$ is always a solution to this equation.
    Furthermore, $X = 0$ is the unique solution if the eigenvalues of $A$ and $B$ are disjoint.
  \end{lemma}
  Combining~\cref{eq:commutator-q0,lem:sylvester} we have the following condition to rule out translation symmetry breaking.
  \begin{proposition}[No translation symmetry breaking]
    \label{prop:no-translation-breaking}
    Suppose that $P$ is the 1-RDM of a Hartree-Fock ground state of a frustration-free FBI Hamiltonian~\cref{eq:h-fbi}.
    For any $\vk, \vk' \in \mc{K}$, if there exists a $\vG \in \Gamma^*$ so that the eigenvalues of $\Lambda_{\vk}(\vG)$ and $\Lambda_{\vk'}(\vG)$ are disjoint, then $P(\vk, \vk') = 0$.
  \end{proposition}

  If the eigenvalues of $A$ and $B$ in~\cref{eq:sylvester} are not disjoint then there can be many possible choices for $X$ which solve $X A - B X = 0$.
  While characterizing all such possible solutions analytically is difficult in the general case, for the special case of TBG-2 there is additional structure which allow us to precisely determine the structure of the degeneracy in solutions to~\cref{eq:commutator-q0}.

  We begin by proving that no translation breaking occurs in spinless, valleyless TBG-2 in~\cref{sec:application-to-spinl} and then characterize the degeneracy in valleyful and spinful models in~\cref{sec:extens-vall-spinf}.
  Finally in~\cref{sec:alg-for-translation-breaking}, we outline a numerical algorithm which can be used to determine a spanning set for the vector space of solutions to~\cref{eq:commutator-q0} and prove its correctness.

  \subsection{Application to Spinless, Valleyless TBG}
  \label{sec:application-to-spinl}
  We begin by converting~\cref{eq:commutator-q0} from an equation depending on $\vG \in \Gamma^*$ to an equation depending on $\vr \in \Omega$ by taking the Fourier transform.

  First, we define $\rho_{\vk,\vk'}(\vr) \in \CC^{(2 M) \times (2 M)}$ as follows
  \begin{equation}
    [\rho_{\vk,\vk'}(\vr)]_{m,n} := \braket{u_{m\vk}(\vr), u_{n\vk}(\vr)}.
  \end{equation}
  Since the periodic Bloch functions, $u_{n\vk}(\vr)$, are periodic with respect to $\Gamma$, $\rho_{\vk,\vk'}(\vr)$ is also periodic with respect to $\Gamma$; that is, $\rho_{\vk,\vk'}(\vr + \vR) = \rho_{\vk,\vk'}(\vr)$ for all $\vR \in \Gamma$. 
  A straightforward calculation  shows that that the function $\vG \mapsto \Lambda_{\vk}(\vq + \vG)$ coincides with the Fourier coefficients of $\rho_{\vk, \vk + \vq}(\vr)$ when considered as a periodic function on the unit cell $\Omega$.
  More specifically, 
  \begin{equation}
    \label{eq:rho-fourier-transform}
    \rho_{\vk,\vk + \vq}(\vr) = \sum_{\vG \in \Gamma^*} e^{-i \vG \cdot \vr} \Lambda_{\vk}(\vq + \vG).
  \end{equation}
  Therefore, by multiplying by $e^{i \vG \cdot \vr}$ and summing over $\vG \in \Gamma^*$ we can transform~\cref{eq:commutator-q0} to real space as follows
  \begin{equation}
    \label{eq:commutator-q0-real}
    P(\vk, \vk') \rho_{\vk',\vk'}(\vr) - \rho_{\vk,\vk}(\vr) P(\vk, \vk') = 0, \qquad \forall \vk, \vk' \in \mc{K}, ~\forall \vr \in \Omega.
  \end{equation}
  The Hamiltonian for TBG-2 in the chiral limit takes the form~\cite{TarnopolskyKruchkovVishwanath2019,BeckerEmbreeWittstenEtAl2022} 
  \begin{equation}
    \label{eq:family}
    H_{\vk} = \begin{bmatrix} 0 & D_{\vk}(\alpha)^{\dagger} \\ D_{\vk}(\alpha) & 0 \end{bmatrix}.
  \end{equation}
  Here, $D_{\vk}(\alpha) = D(\alpha)+ (k_{x} + i k_{y})I_{2 \times 2}$ where $D(\alpha) = \begin{bmatrix} 
    2\bar{D} & \alpha V(\vr) \\
    \alpha V(-\vr) & 2 \bar{D}  
  \end{bmatrix}$ with $\overline{D}=-i \partial_{\bar z}$ and $V(\vr)= \sum_{n=0}^2 \omega^n e^{-i \vq_n \cdot \vr}$ with $\omega=e^{2\pi i /3}$. 

  The Hamiltonian $H_{\vk}$ satisfies two symmetries $\mc{Q}$ and $\mc{L}$ whose actions are defined as follows
  \begin{equation}
  \label{eq:symmetries}
          \mc{Q} u(\vr) := \begin{bmatrix}
              & & 1 & \\
              & & & 1 \\
              1 & & & \\
              & 1 & & 
          \end{bmatrix} \overline{u(-\vr)} \qquad
          \mc{L} u(\vr) := \begin{bmatrix}
              & 1 & & \\
             -1 & & & \\
               & & & 1 \\
              & & -1 & 
          \end{bmatrix} \overline{u(-\vr)} \\
  \end{equation}
  Note that $\mc{Q}$ maps $\vk$ to $\vk$ and $\mc{L}$ maps $\vk$ to $-\vk$ and $[\mc{Q},\mc{L} ] = 0$.
  
  Due to $\mc{Q}$ symmetry, the zero eigenvectors of $H_{\vk}$ can be written as $\begin{bmatrix} 0 & v(\vr) \end{bmatrix}^{\top}$ and $\begin{bmatrix} \,\overline{v(-\vr)} & 0 \end{bmatrix}^{\top}$.
  Hence after a basis rotation, $\rho_{\vk,\vk}(\vr)$ can be written as follows
  \begin{equation}
    \label{eq:pair-product-tbg2}
    \rho_{\vk,\vk}(\vr) =
    \begin{bmatrix}
      \| u_{\vk}(\vr) \|^2 & \\
                           & \| u_{\vk}(-\vr) \|^2
    \end{bmatrix}.
  \end{equation}
  Due to $\mc{L}$ symmetry and the fact that $[\mc{Q}, \mc{L} ] = 0$, we may also choose an appropriate basis change so that a state $u_{\mathbf 0}$ satisfying $D(\alpha)u_{\mathbf 0}=0$ has the property
  \[ \left\Vert u_{\mathbf 0} (\vr)  \right\Vert =\left\Vert  \mathcal Lu_{\mathbf 0}(\vr) \right\Vert=\left\Vert  u_{\mathbf 0}(-\vr) \right\Vert,   \]
  where, in this particular case, $\mathcal L$ just denotes the upper $2\times 2$ block of the symmetry defined in \eqref{eq:symmetries}.
    With these observations, we can now characterize all Hartree-Fock ground states of the FBI model for TBG-2.
  \begin{theorem}
    \label{thm:tbg-2}
    For TBG-2 at half-filling without the valley or spin degrees of freedom, the Hartree-Fock ground states take the form
    \begin{equation}
      P(\vk, \vk')
      =
      \begin{cases}
        P_{0} & \vk = \vk' \\
        0 & \vk \neq \vk'
      \end{cases}
    \end{equation}
    where $P_{0}$ the same for all $\vk \in \mc{K}$ and equal to either $\diag(1,0)$ or $\diag(0,1)$.
  \end{theorem}
  \begin{remark}
    The states found in~\cref{thm:tbg-2} are the same ferromagnetic Slater determinant states found in our previous work \cite{BeckerLinStubbs2023}.
  \end{remark}
  \begin{proof}
    By \cite[Lemma $3.1$]{BeckerHumbertZworski2022a}, we know that at a magic $\alpha$ the one-particle Hamiltonian $H$ exhibits precisely two flat bands at energy zero
    \begin{equation}
      \label{eq:identity}
      \vu_{\vk}(\vr) = F_{\vk}(\vr) \vu_{\mathbf 0}(\vr) \in \ker(D(\alpha)+(k_1+ik_2)I_{2\times 2}),
    \end{equation}
    where
    \[
      F_{\vk}(\vr):=e^{\frac{(k_1+ik_2)}{2}(-i(1+\omega)x+(\omega-1)y)} \frac{ \theta ( \frac{3(x_1+ix_2)}{4\pi i \omega} + \frac{k_1+ik_2}{\sqrt{3}\omega} ) }{
        \theta\Big( \frac{3(x_1+ix_2)}{4\pi i \omega}\Big)}.
    \]
    Since $\theta$ is an odd function, $F_{\vk}(\vr)$ satisfies the reflection symmetry 
    \[
      F_{\vk}(\vr)=F_{-\vk}(-\vr).
    \]
    The function $\vu_{\vk}(\vr)$ has a unique zero at $\vr = \frac{4\pi}{3\sqrt{3}}(k_2,-k_1)^\top$ per unit cell (note that both $\vu_{\mathbf 0}$ and $\theta$ have a simple zero at $\vr=\mathbf 0$). 
    This implies that 
    \[
      \Vert u_{\vk}(\vr)\Vert =\Vert F_{\vk}(\vr) u_{\mathbf 0}(\vr)\Vert =\Vert F_{-\vk}(-\vr) u_{\mathbf 0}(-\vr)\Vert =  \Vert u_{-\vk}(-\vr)\Vert.
    \]

    To show that $P(\vk, \vk') =0$ for a disjoint pair $\vk,\vk' \in \mathcal K$, it suffices to show by  \eqref{eq:commutator-q0-real}, that there are $\vr,\vr' \in \Omega$ such that 
    \[ {\| u_{\vk}(\vr) \| \neq \| u_{\vk'}(\vr) \|} \text{ and }\| u_{\vk}(\vr') \|  \neq \| u_{\vk'}(-\vr') \|.  \]
    To satisfy the first constraint, we may choose the unique point $\vr = \frac{4\pi}{3\sqrt{3}}(k_2,-k_1)^\top$ at which $u_{\vk}(\vr)=0,$ then $u_{\vk'}(\vr)\neq 0.$

    For this choice of $\vr,$ we also have 
    \[u_{\vk'}(-\vr) \neq 0 \text{ as long as }\vk' \neq -\vk. \]
    This implies that 
    \[ P(\vk,\vk')=0\text{ for }\vk' \neq -\vk.\]

While we cannot directly conclude that $P(\vk,-\vk)=0$ for $\vk \notin \Gamma^*$, we may invoke \cref{eq:symmetry}, and use that we have previously shown that $P(\vk, \vk') = 0$ whenever $\vk \notin \{\vk',-\vk'\}$.
In particular, \cref{eq:symmetry} implies that for any $\vq'\in \R^2 \setminus \Gamma^*$,
\[
  \operatorname{rank}(P(\vk,-\vk))=\operatorname{rank}(P(\vk+\vq',-\vk+\vq'))=0.
\]
This implies that
\[P(\vk, -\vk) = 0\text{ for all }\vk \in \mc{K} \setminus \{ \vzero \}.\]

Finally, for $\vk=\vk'\neq \mathbf 0$ it follows from specializing in \eqref{eq:pair-product-tbg2} to $\vr = \frac{4\pi}{3\sqrt{3}}(k_2,-k_1)^\top$, in which case $\rho_{\vk,\vk}(\vr)= \operatorname{diag}(0,\Vert u_{\vk}(-\vr)\Vert^2)$ with $u_{\vk}(-\vr) \neq 0.$ 
This implies that $P(\vk,\vk)$ is diagonal.
Since $P^{2} = P$ and $P(\vk, \vk') = 0$ whenever $\vk \neq \vk'$, we conclude that $P(\vk,\vk)^2 =P(\vk,\vk)$ and so its eigenvalues are $0$ and $1$.
At half-filling, by \cref{coro:uniform-filling}, we know $\Tr{(P(\vk, \vk))} = 1$ and so $P(\vk,\vk)=\operatorname{diag}(1,0)$ or $P(\vk,\vk)=\operatorname{diag}(0,1)$ for all $\vk$.

Finally, for $\vk'=\vk=\mathbf 0$, the condition \eqref{eq:pair-product-tbg2} does not imply any restrictions on $P(\mathbf 0,\mathbf 0),$ since $\rho_{\mathbf 0,\mathbf 0}$ is proportional to the identity matrix. 
However, the proof of~\cref{prop:pseudo-ferromagnetism} shows that one can choose diagonal $B_{\vk}(\mathbf q')$. This implies by \eqref{eq:pair-product-tbg2} that $\vk \mapsto P(\vk,\vk)$ is either $\operatorname{diag}(1,0)$ or $\operatorname{diag}(0,1)$ for all $\vk$.

To complete the proof, we must verify that~\cref{eq:hf-trace-condition} also vanishes.
This follows from the sum rule
\begin{equation}
  \label{eq:sum-rule}
  \sum_{\vk} \Tr{\left(\Lambda_{\vk}(\vG) \begin{bmatrix} 1 & \\ & -1 \end{bmatrix}\right)} = 0 \quad \forall \vG \in \Gamma^{*}
\end{equation}
proven in our previous work~\cite[Lemma 4.3]{BeckerLinStubbs2023}.

\end{proof}

\subsection{Extension to Valleyful and Spinful Models}
\label{sec:extens-vall-spinf}
In the previous section, we showed that in single valley TBG-2 the only possible Hartree-Fock ground states are the two ferromagnetic Slater determinants by using the specific properties of the TBG-2 eigenfunctions.
While the final conclusion depended on the specific form of these eigenfunctions all of the steps up to and including~\cref{eq:commutator-q0-real} were done in a completely general setting (i.e. without reference to an underlying Hamiltonian).
Therefore, to translate these results to the more realistic valleyful and spinful settings we simply need to substitute in appropriate definition for $\rho_{\vk,\vk}(\vr)$.

Following the analysis in \cite{BultinckKhalafLiuEtAl2020}, for the chiral model of TBG-2 with valley, the corresponding $\rho_{\vk,\vk}(\vr)$ takes the form:
\begin{equation}
  \rho_{\vk,\vk}^{\rm V}(\vr) =
  \begin{bmatrix}
    \| u_{\vk}(\vr) \|^2 &&& \\
                         & \| u_{\vk}(-\vr) \|^2 && \\
                         && \| u_{\vk}(-\vr) \|^2 & \\
                         &&& \| u_{\vk}(\vr) \|^2
  \end{bmatrix}
\end{equation}
and when including valley and spin it takes the form
\begin{equation}
  \rho_{\vk,\vk}^{\rm V+S}(\vr) =
  \begin{bmatrix}
    \| u_{\vk}(\vr) \|^2 &&& \\
                         & \| u_{\vk}(-\vr) \|^2 && \\
                         && \| u_{\vk}(-\vr) \|^2 & \\
                         &&& \| u_{\vk}(\vr) \|^2
  \end{bmatrix} \otimes I_{2 \times 2}.
\end{equation}
We will now show that when including the valley degree of freedom, the ground state exhibits a ${\rm U}(2) \times {\rm U}(2)$ degeneracy.
A similar argument shows that in the presence of the valley \textit{and spin} degrees of freedom, the ground state is unique up to a ${\rm U}(4)\times {\rm U}(4)$ degree of freedom.

\begin{theorem}
  \label{thm:tbg-2-valley}
  At TBG-2 at half-filling with the valley degree of freedom but without spin, the Hartree-Fock ground states take the form
  \begin{equation}
    P(\vk, \vk')
    =
    \begin{cases}
      U P_{0} U^{\dagger} & \vk = \vk' \\
      0 & \vk \neq \vk'
    \end{cases}
  \end{equation}
  where $U$ and $P_{0}$ are the same for all $\vk \in \mc{K}$.
  Furthermore, $U \in {\mathrm U}(2) \times {\mathrm U}(2)$ and $P_{0}$ is  one of the following matrices $\diag(1,0,0,1)$, $\diag(1,1,0,0)$, and $\diag(0, 1, 1, 0)$.
\end{theorem}
\begin{proof}
  Following the same reasoning used in the proof of~\cref{thm:tbg-2}, we know that $P(\vk, \vk') = 0$ for all $\vk \neq \vk'$.
  By \cref{prop:pseudo-ferromagnetism}, it suffices to choose a single $\vk \in \mc{K}$ to fix the ground state.
  Pick any $\vk_{*} \in \mc{K}$ so that $\| u_{\vk_{*}}(\vr) \| \neq \| u_{\vk_{*}}(-\vr) \|$ for some $\vr \in \Omega$.
  To make the algebra simpler, we define the permutation matrix $\Pi$:
  \begin{equation}
    \Pi =
    \begin{bmatrix}
      1 &&& \\
        &&& 1 \\
        && 1 & \\
        & 1 && 
    \end{bmatrix}
  \end{equation}
  which satisfies $\Pi^{\dagger} = \Pi$ and $\Pi^{2} = I$.
  Applying $\Pi$ to both sides of~\cref{eq:pair-product-tbg2} we find that for $P(\vk_{*},\vk_{*})$ to be the 1-RDM of a ground state it must be that
  \begin{equation}
    \label{eq:commutator-q0-real-pi}
    \Big( \Pi P(\vk_{*}, \vk_{*}) \Pi \Big) \Big( \Pi \rho_{\vk_{*},\vk_{*}}(\vr) \Pi \Big) - \Big( \Pi \rho_{\vk_{*},\vk_{*}}(\vr) \Pi \Big) \Big( \Pi P(\vk_{*}, \vk_{*}) \Pi \Big) = 0.
  \end{equation}
  We will derive a characterization for all possible solutions $\Pi P(\vk_{*}, \vk_{*}) \Pi$ of \cref{eq:commutator-q0-real-pi} at half-filling.
  Since $\Pi$ is a unitary, this also characterizes all possible solutions $P(\vk_{*}, \vk_{*})$ satisfying~\cref{eq:commutator-q0-real} after a change of basis.

  Notice that
  \begin{equation}
  \label{eq:pair-product-form}
    \begin{split}
      \Pi \rho_{\vk_{*},\vk_{*}}^{\rm V}(\vr) \Pi & =
      \begin{bmatrix}
        \| u_{\vk_{*}}(\vr) \|^2 &&& \\
                             & \| u_{\vk_{*}}(\vr) \|^2 && \\
                             && \| u_{\vk_{*}}(-\vr) \|^2 & \\
                             &&& \| u_{\vk_{*}}(-\vr) \|^2
      \end{bmatrix} \\[1ex]
      & =
      \begin{bmatrix}
        \| u_{\vk_{*}}(\vr) \|^{2} I_{2 \times 2} & \\
                                                & \| u_{\vk_{*}}(-\vr) \|^{2} I_{2 \times 2}
      \end{bmatrix}.
    \end{split}
      \end{equation}
    So splitting $\Pi P(\vk_{*}, \vk_{*}) \Pi$ into $2 \times 2$ blocks
  \begin{equation}
    \Pi P(\vk_{*}, \vk_{*}) \Pi =:
    \begin{bmatrix}
      \tilde{P}_{11} & \tilde{P}_{12}  \\
      \tilde{P}_{12}^{\dagger} & \tilde{P}_{22}
    \end{bmatrix},
  \end{equation}
  we can explicitly compute the commutator in terms of $\tilde{P}_{11}$, $\tilde{P}_{22}$, $\tilde{P}_{12}$:
  \begin{equation}
    [ \Pi P(\vk_{*}, \vk_{*}) \Pi, \Pi \rho_{\vk_{*},\vk_{*}}^{\rm V}(\vr) \Pi ]
    =
    (\| u_{\vk_{*}}(-\vr) \|^{2} - \| u_{\vk_{*}}(\vr) \|^{2})
    \begin{bmatrix}
      & \tilde{P}_{12} \\
      - \tilde{P}_{12}^{\dagger} & 
    \end{bmatrix}.
  \end{equation}
  Since $\| u_{\vk_{*}}(-\vr) \|^{2} \neq \| u_{\vk_{*}}(\vr) \|^{2}$ for some $\vr$, for \cref{eq:pair-product-tbg2} to vanish it must be that $\tilde{P}_{12} \equiv 0$.
  Hence any solution to~\cref{eq:commutator-q0-real-pi} must be block diagonal.

  Since $P(\vk_{*},\vk') = 0$ for all $\vk_{*} \neq \vk'$, $\Pi P(\vk_{*}, \vk_{*}) \Pi$ is an orthogonal projection.
  The set of all transformations preserving this block structure while maintaining orthogonality is given by $U \in {\rm U}(2) \times {\rm U}(2)$ where $U$ takes the form
  \begin{equation}
    U =
    \begin{bmatrix}
      U_{1} & \\
            & U_{2}
    \end{bmatrix} \quad U_{1}, U_{2} \in {\rm U}(2).
  \end{equation}
  Due to \cref{coro:uniform-filling}, at half-filling, $\Pi P(\vk_{*}, \vk_{*}) \Pi$ has rank 2 and hence all possible solutions to~\cref{eq:commutator-q0-real-pi} must lie in one of three orbits:
  \begin{equation}
  \label{eq:valley-gs-rdm}
    U
    \begin{bmatrix}
      1 &&& \\
        & 1 && \\
        && 0 & \\
        &&& 0
    \end{bmatrix}
    U^{\dagger}
    \qquad
    U
    \begin{bmatrix}
      1 &&& \\
        & 0 && \\
        && 1 & \\
        &&& 0
    \end{bmatrix}
    U^{\dagger}
    \qquad
    U
    \begin{bmatrix}
      0 &&& \\
        & 0 && \\
        && 1 & \\
        &&& 1
    \end{bmatrix}.
    U^{\dagger}
  \end{equation}
Suppose we choose $\Pi P(\vk_{*}, \vk_{*}) \Pi$ to be one of the options in~\cref{eq:valley-gs-rdm}; we now will argue that $\Pi P(\vk, \vk) \Pi = \Pi  P(\vk_{*}, \vk_{*}) \Pi $ for all $\vk \in \mc{K}$.
Due to $\mc{Q}$ symmetry, the pair product $\Pi \rho_{\vk,\vk+\vq}^{V}(\vr) \Pi$ has the same form as in~\cref{eq:pair-product-form} for all $\vk, \vk + \vq$~\cite{BultinckKhalafLiuEtAl2020}. Therefore, the pair product commutes with $\Pi  P(\vk_{*}, \vk_{*}) \Pi$ and so following the proof of~\cref{prop:pseudo-ferromagnetism}, we can conclude for a ground state $\Pi P(\vk, \vk) \Pi = \Pi  P(\vk_{*}, \vk_{*}) \Pi$ for all $\vk \in \mc{K}$.
  
  To complete the proof, we must also verify that that~\cref{eq:hf-trace-condition} vanishes.
  Since $[ U, \Pi \rho_{\vk,\vk}^{V}(\vr) \Pi ] = 0$ and $\Pi \rho_{\vk,\vk}^{V}(\vr) \Pi = I_{2 \times 2} \otimes \rho_{\vk,\vk}(\vr)$, the sum rule~\cref{eq:sum-rule} implies~\cref{eq:hf-trace-condition} vanishes.
The statement of the theorem follows reordering the basis by the permutation $\Pi$.
\end{proof}

\begin{remark}
We shall compare our result with statements in the physics literature.
  \begin{enumerate}
  \item In the physics literature, the structure~\cref{eq:commutator-q0-real-pi} is often expressed as $[\sigma_z \tau_z,P(\vk,\vk)]=0$ where $\sigma_{z}$ denotes the Pauli Z-matrix acting on the sublattice degree of freedom and $\tau_{z}$ denotes the Pauli Z-matrix acting on the valley degree of freedom.
    The operator $\sigma_z \tau_z$ is called the Chern number operator.
  \item In the physics literature, the ${\rm U}(2)\times {\rm U}(2)$ (or  ${\rm U}(4)\times {\rm U}(4)$) degree of freedom is often stated as a hidden symmetry of the FBI Hamiltonian.
    The two perspectives are equivalent.
    This is because if $\ket{\Psi}$ is a ground state of $\hat{H}_{\rm FBI}$ and $\mc{G}$ is a symmetry satisfying $[\mc{G}, \hat{H}_{FBI}]=0$, then $\mc{G}\ket{\Psi}$ is also a ground state.

  \item Since $U$ is block diagonal, the orbits of $\diag(1, 0, 0, 1)$ and $\diag(0, 1, 1, 0)$ only contain one element.
    These two states are referred to as the quantum hall (QH) states as they occupy the positive eigenspace and negative eigenspace of the Chern number operator $\sigma_{z} \tau_{z}$.
    In contrast, the third orbit contains a manifold of states which include states which occupy both valleys; such states are referred to as the intervalley coherent (IVC) states.
  \end{enumerate}
\end{remark}

Using the same proof strategy, we can easily extend~\cref{thm:tbg-2-valley} to include spin:
\begin{theorem}
  \label{thm:tbg-2-valley-spin}
  At TBG-2 at half-filling with the valley and spin degrees of freedom, the Hartree-Fock ground states take the form
  \begin{equation}
    P(\vk, \vk')
    =
    \begin{cases}
      U P_{0} U^{\dagger} & \vk = \vk' \\
      0 & \vk \neq \vk'
    \end{cases}
  \end{equation}
  where $U$ and $P_{0}$ are the same for all $\vk \in \mc{K}$.
  Furthermore, $U \in {\mathrm U}(4) \times {\mathrm U}(4)$ and $P_{0}$ is equal to one of the following matrices $\diag(1,0,0,1, 1,0,0,1)$, $\diag(1,1,0,1, 1,0,0,0)$, $\diag(1,1,1,1,0,0,0,0)$, $\diag(0,0,1,0, 0,1,1,1)$, or $\diag(0,1,1,0, 0,1,1,0)$.
\end{theorem}

\subsection{An Algorithm for Determining Translation Symmetry Breaking}
\label{sec:alg-for-translation-breaking}
In cases when an analytical form of the form factor is not known,~\cref{eq:commutator-q0} can be used to design an algorithm for identifying when translation symmetry breaking can occur.
From the vectorization identity, we can convert \cref{eq:commutator-q0} into a vector equation as follows
\begin{equation}
  \label{eq:vec-commutator-q0}
  \left( \Lambda_{\vk'}(\vG)^{\top} \otimes I_{(2M) \times (2M)} - I_{(2M) \times (2M)} \otimes \Lambda_{\vk}(\vG) \right) \operatorname{vec}{(P(\vk, \vk'))} = 0 \quad \forall \vG \in \Gamma^*,
\end{equation}
where $\operatorname{vec}{(\cdot)}$ transforms a $(2M) \times (2M)$ matrix into a vector with $(2M)^{2}$ entries by stacking the columns.
Define the $(2M)^{2} \times (2M)^{2}$ matrix $M_{\vk,\vk'}(\vG)$:
\begin{equation}
  \label{eq:mkk_g}
  M_{\vk,\vk'}(\vG) := \Lambda_{\vk'}(\vG)^{\top} \otimes I_{(2M) \times (2M)} - I_{(2M) \times (2M)} \otimes \Lambda_{\vk}(\vG).
\end{equation}
Then $P(\vk, \vk')$ satisfies~\cref{eq:commutator-q0} if and only if its vectorization lies in the kernel $M_{\vk,\vk'}(\vG)$ for all $\vG \in \Gamma^{*}$.
We can further reduce our computational effort by squaring $M_{\vk,\vk'}(\vG)$ and summing over $\vG$:
\begin{equation}
  \label{eq:vec-matrix}
  M_{\vk, \vk'} := \sum_{\vG \in \Gamma^{*}} M_{\vk,\vk'}(\vG)^{\dagger} M_{\vk,\vk'}(\vG).
\end{equation}
The following proposition is now immediate
\begin{proposition}
  \label{prop:translation-breaking-alg}
  The matrix $M_{\vk, \vk'}$ is positive semidefinite and its kernel is exactly the intersection of the kernels of all the $M_{\vk,\vk'}(\vG)$.
\end{proposition}
Since the form factors are the Fourier coefficients of a smooth function (\cref{eq:rho-fourier-transform}) we have that
\begin{equation}
  \sum_{\{ \vG \in \Gamma^* : |\vG| > L \}} \| \Lambda_{\vk'}(\vG) \|_{F}^{2} \rightarrow 0 \quad \text{as} \quad L \rightarrow \infty
\end{equation}
Hence, in practice we can truncate the sum over $\vG$ in~\cref{eq:vec-matrix} and by the Weyl bound we can find a spanning set for all $P(\vk, \vk')$ lying in the kernel of $M_{\vk,\vk'}$ by taking a few of the smallest eigenmodes of the truncated matrix.
\begin{remark}
  In a discretized setting, \cref{prop:translation-breaking-alg} also gives us a computationally efficient method for verifying the ``no common invariant subspace'' condition in \cite[Theorem 4]{BeckerLinStubbs2023}.
\end{remark}

\bibliographystyle{plain}
\bibliography{bibliography}

\appendix

\section{Calculation of the Hartree-Fock Energy}
\label{sec:calc-hf-energy}
To perform this calculation, we will recall a few properties of the form factor which will play a role in the following calculations (see, e.g. \cite[Lemma 4.2]{BeckerLinStubbs2023} for proofs):
\begin{lemma}
  \label{lem:form-factor-identities}
  The form factor matrix $\Lambda_{\vk}(\vq + \vG)$ satisfies the following identities for all $\vk,\vq \in \mc{K}$ and all $\vG, \vG' \in \Gamma^*$
  \begin{align}
    \Lambda_{\vk}(\vq + \vG)^{\dagger} &  = \Lambda_{\vk+\vq}(-\vq -\vG), \label{eq:form-factor-dagger} \\
    \Lambda_{\vk + \vG'}(\vq + \vG) & = \Lambda_{\vk}(\vq + \vG). \label{eq:form-factor-shift}
  \end{align}
\end{lemma}
Throughout this calculation, we will suppress the summation indexes for $\vq', \vk, \vk', m, n, m', n'$.
Expanding the terms in $\widehat{\rho}(\vq')$ and performing normal ordering gives 
\begin{equation}
  \label{eq:h-fbi-alt}
  \begin{split}
    N_{\vk} |\Omega| \hat{H}_{\rm FBI} = 
    & \sum_{\vq'} \sum_{\vk, \vk'} \sum_{m,n,m',n'} \hat{V}(\vq') [\Lambda_{\vk}(\vq')]_{mn} [\Lambda_{\vk'}(-\vq')]_{m'n'} \hat{f}_{m\vk}^{\dagger} \hat{f}_{m'\vk'}^{\dagger} \hat{f}_{n'(\vk' - \vq')} \hat{f}_{n(\vk + \vq')}  \\[1ex]
    & + \sum_{\vq'}  \sum_{\vk'} \sum_{m,n} \hat{V}(\vq') [\Lambda_{\vk + \vq'}(\vq') \Lambda_{\vk'}(-\vq')]_{mn} \hat{f}_{m\vk'}^{\dagger} \hat{f}_{n\vk'} \\
    & - \frac{1}{2} \sum_{\vq'} \delta_{\vq' \in \Gamma^*} \hat{V}(\vq') \left( \sum_{\vk'} \Tr{(\Lambda_{\vk'}(-\vq'))} \right) \left( \sum_{\vk} \sum_{m,n} [\Lambda_{\vk}(\vq')]_{mn} \hat{f}_{m\vk}^{\dagger} \hat{f}_{n\vk} \right) \\
    & - \frac{1}{2} \sum_{\vq'} \delta_{\vq' \in \Gamma^*} \hat{V}(\vq') \left( \sum_{\vk} \Tr{(\Lambda_{\vk}(\vq'))} \right) \left( \sum_{\vk'} \sum_{m',n'} [\Lambda_{\vk'}(-\vq')]_{mn} \hat{f}_{m'\vk'}^{\dagger} \hat{f}_{n'\vk'} \right) \\
    & + \frac{1}{4} \sum_{\vq'} \delta_{\vq' \in \Gamma^* } \hat{V}(\vq') \left(\sum_{\vk} \Tr{(\Lambda_{\vk}(\vq'))} \right) \left( \sum_{\vk'} \Tr{(\Lambda_{\vk'}(-\vq'))} \right).
  \end{split}
\end{equation}

By Wick's theorem, we can evaluate the energy of any Hartree-Fock state, $\ket{\Psi}$ in terms of the one-body reduced density matrix $[P(\vk',\vk)]_{nm} := \braket{\Phi | \hat{f}_{m\vk}^{\dagger} \hat{f}_{n\vk'} | \Phi}$.
In particular, we have
\begin{equation}
  \begin{split}
    N_{\vk} |\Omega| & \langle\Psi | \hat{H}_{\rm FBI} | \Psi \rangle \\
    = & \sum_{\vq'} \hat{V}(\vq') \sum_{\vk, \vk'} \sum_{m,n,m',n'} [\Lambda_{\vk}(\vq')]_{mn} [\Lambda_{\vk'}(-\vq')]_{m'n'} [P(\vk + \vq', \vk)]_{nm} [P(\vk' - \vq', \vk')]_{n'm'} \\
                     & - \sum_{\vq'} \hat{V}(\vq') \sum_{\vk, \vk' } \sum_{m,n,m',n'} [\Lambda_{\vk}(\vq')]_{mn} [\Lambda_{\vk'}(-\vq')]_{m'n'} [P(\vk' -\vq', \vk)]_{n'm} [P(\vk + \vq', \vk')]_{nm'} \\
                     & - \sum_{\vq'} \hat{V}(\vq') \sum_{\vk'} \sum_{m,n} [\Lambda_{\vk' + \vq'}(\vq') \Lambda_{\vk'}(-\vq')]_{mn} [P(\vk',\vk')]_{nm} \\
                     & - \frac{1}{2} \sum_{\vq'} \delta_{\vq' \in \Gamma^*} \hat{V}(\vq') \left( \sum_{\vk'} \Tr{(\Lambda_{\vk'}(-\vq'))} \right) \left( \sum_{\vk} \sum_{m,n} [\Lambda_{\vk}(\vq')]_{mn} [P(\vk,\vk)]_{nm} \right) \\
                     & - \frac{1}{2} \sum_{\vq'} \delta_{\vq' \in \Gamma^*} \hat{V}(\vq') \left( \sum_{\vk} \Tr{(\Lambda_{\vk}(\vq'))} \right) \left( \sum_{\vk'} \sum_{m',n'} [\Lambda_{\vk'}(-\vq')]_{mn} [P(\vk',\vk')]_{n'm'} \right) \\
                     & + \frac{1}{4} \sum_{\vq'} \delta_{\vq' \in \Gamma^{*}} \hat{V}(\vq') \left(\sum_{\vk} \Tr{(\Lambda_{\vk}(\vq'))} \right) \left( \sum_{\vk'} \Tr{(\Lambda_{\vk'}(-\vq'))} \right)
  \end{split}
\end{equation}
which can cleanly be expressed in terms of traces as
\begin{equation}
  \begin{split}
    N_{\vk} |\Omega| & \langle\Psi | \hat{H}_{\rm FBI} | \Psi \rangle \\
    = &  \sum_{\vq'} \hat{V}(\vq') \sum_{\vk, \vk'} \Tr{\Big( \Lambda_{\vk}(\vq') P(\vk + \vq', \vk) \Big)} \Tr{\Big( \Lambda_{\vk'}(-\vq') P(\vk' - \vq', \vk') \Big)} \\
                     & -  \sum_{\vq'} \hat{V}(\vq') \sum_{\vk, \vk'} \Tr{\Big( \Lambda_{\vk}(\vq') P(\vk + \vq', \vk') \Lambda_{\vk'}(-\vq') P(\vk' -\vq', \vk)\Big)} \\
                     & - \sum_{\vq'} \hat{V}(\vq') \sum_{\vk'} \Tr{\Big(\Lambda_{\vk' + \vq'}(\vq') \Lambda_{\vk'}(-\vq') P(\vk',\vk') \Big)} \\
                     & - \frac{1}{2} \sum_{\vq'} \delta_{\vq' \in \Gamma^{*}}  \hat{V}(\vq') \left( \sum_{\vk'} \Tr{(\Lambda_{\vk'}(-\vq'))} \right) \left( \sum_{\vk} \Tr{(\Lambda_{\vk}(\vq') P(\vk,\vk))} \right) \\
                     & - \frac{1}{2} \sum_{\vq'} \delta_{\vq' \in \Gamma^{*}}  \hat{V}(\vq') \left( \sum_{\vk} \Tr{(\Lambda_{\vk}(\vq'))} \right) \left( \sum_{\vk'} \Tr{(\Lambda_{\vk'}(-\vq') P(\vk',\vk'))} \right) \\
                     & + \frac{1}{4} \sum_{\vq'}  \delta_{\vq' \in \Gamma^{*}} \hat{V}(\vq') \left(\sum_{\vk} \Tr{(\Lambda_{\vk}(\vq'))} \right) \left( \sum_{\vk'} \Tr{(\Lambda_{\vk'}(-\vq'))} \right).
  \end{split}
\end{equation}
Recalling the definition of the matrix $Q(\vk, \vk')$ (\cref{eq:q-def}), we can combine the first term and the last three terms to get 
\begin{equation}
  \begin{split}
    N_{\vk} |\Omega| \braket{\Psi | \hat{H}_{\rm FBI} | \Psi} =
    &  \sum_{\vq'} \hat{V}(\vq') \sum_{\vk, \vk'} \Tr{\Big( \Lambda_{\vk}(\vq') Q(\vk + \vq', \vk) \Big)} \Tr{\Big( \Lambda_{\vk'}(-\vq') Q(\vk' - \vq', \vk') \Big)} \\
    & -  \sum_{\vq'} \hat{V}(\vq') \sum_{\vk, \vk'} \Tr{\Big( \Lambda_{\vk}(\vq') P(\vk + \vq', \vk') \Lambda_{\vk'}(-\vq') P(\vk' -\vq', \vk)\Big)} \\
    & - \sum_{\vq'} \hat{V}(\vq') \sum_{\vk'} \Tr{\Big(\Lambda_{\vk' + \vq'}(\vq') \Lambda_{\vk'}(-\vq') P(\vk',\vk') \Big)}.
  \end{split}
\end{equation}
Using~\cref{eq:form-factor-dagger} and performing a change of variables $\vk' \mapsto \vk' - \vq'$ (this change of variables is valid due to \cref{eq:form-factor-shift}), we can combine the last two terms in the above equation to get
\begin{equation}
  \begin{split}
    N_{\vk} |\Omega|  \langle \Psi & | \hat{H}_{\rm FBI} | \Psi \rangle \\
    = & \sum_{\vq'} \hat{V}(\vq')  \sum_{\vk, \vk' \in \mc{K}} \Tr{\Big( \Lambda_{\vk}(\vq') Q(\vk + \vq', \vk) \Big)} \Tr{\Big( \Lambda_{\vk'}(-\vq') Q(\vk' - \vq', \vk') \Big)} \\
    & - \sum_{\vq'} \hat{V}(\vq') \sum_{\vk, \vk' \in \mc{K}} \Tr{\Big( \Lambda_{\vk}(\vq') Q(\vk + \vq', \vk') \Lambda_{\vk'}(-\vq') Q(\vk' -\vq', \vk) \Big)} \\
    & + \frac{1}{4} \sum_{\vq'} \hat{V}(\vq') \sum_{\vk} \Tr{\Big(\Lambda_{\vk}(\vq') \Lambda_{\vk}(\vq')^{\dagger} \Big)}. 
  \end{split}
\end{equation}
Finally, using the identity $\Lambda_{\vk'}(-\vq') = \Lambda_{\vk'-\vq'}(\vq')^{\dagger}$ and performing the change of variables in the first sum $\vk' \mapsto \vk' + \vq'$, we can finally write the energy of $\ket{\Psi}$ as
\begin{equation}
  \label{eq:hf-energy3}
  \begin{split}
    \langle \Psi  &| \hat{H}_{\rm FBI} | \Psi \rangle \\
    = & \frac{1}{N_{\vk} |\Omega|} \sum_{\vq' \in \Gamma^{*} + \mc{K}} \hat{V}(\vq') \left| \sum_{\vk \in \mc{K}} \Tr{\Big( \Lambda_{\vk}(\vq') Q(\vk + \vq', \vk) \Big)} \right|^{2} \\
    & + \frac{1}{4 N_{\vk} |\Omega|} \sum_{\vq' \in \Gamma^{*} + \mc{K}} \hat{V}(\vq') \sum_{\vk \in \mc{K}} \| \Lambda_{\vk}(\vq') \|_{F}^{2} \\
    & - \frac{1}{N_{\vk} |\Omega|} \sum_{\vq' \in \Gamma^{*} + \mc{K}} \hat{V}(\vq') \sum_{\vk, \vk' \in \mc{K}} \Tr{\Big( \Lambda_{\vk}(\vq') Q(\vk + \vq', \vk') \Lambda_{\vk'-\vq'}(\vq')^{\dagger} Q(\vk' -\vq', \vk) \Big)} 
  \end{split}
\end{equation}
which completes the calculation.

\section{Proof of~\cref{lem:trace-lemma}}
\label{sec:trace-lemma-proof}
First, we calculate
\begin{equation}
  \| [A, B] \|_{F}^{2} = \Tr{( [A, B] ([A, B])^{\dagger} )} = -\Tr{( [A, B] ([A^{\dagger}, B^{\dagger}]) )}.
\end{equation}
Now we expand the product of commutators
\begin{equation}
  \begin{split}
    [A, B] [ A^{\dagger}, B^{\dagger} ]
    & = (A B - B A)(A^{\dagger} B^{\dagger} - B^{\dagger} A^{\dagger})                                                                        \\
    & = A B A^{\dagger} B^{\dagger} - A B B^{\dagger} A^{\dagger} - B A  A^{\dagger} B^{\dagger} + B A B^{\dagger} A^{\dagger}                \\
    & = \left( A B A^{\dagger} B^{\dagger} + (A B A^{\dagger} B^{\dagger})^{\dagger} \right) - \left( A B B^{\dagger} A^{\dagger} + B A  A^{\dagger} B^{\dagger} \right).
  \end{split}
\end{equation}
Taking the trace of both sides let's us conclude that
\begin{equation}
  -\| [A, B] \|_{F}^{2} = 2 \Re{\left(\Tr{(A B A^{\dagger} B^{\dagger})}\right)} - \Tr{\left(A^{\dagger} A B B^{\dagger} + A  A^{\dagger} B^{\dagger} B\right)}. 
\end{equation}
Rearranging we conclude that
\begin{equation}
  \Re{\left(\Tr{(A B A^{\dagger} B^{\dagger})}\right)} = \frac{1}{2} \Tr{\Big(AA^{\dagger} B^{\dagger} B + A^{\dagger} A B B^{\dagger}\Big)} - \frac{1}{2} \| [A, B] \|_{F}^{2}
\end{equation}
which proves the lemma.

\end{document}

%% file: preamble.tex
\usepackage{xcolor,colortbl}
\usepackage{graphicx}
\usepackage{geometry}
\usepackage{enumitem}
\usepackage{amsmath,amssymb,amsthm,amsfonts}
\usepackage{algorithm}
\usepackage{algorithmic}
\usepackage{dcolumn}
\usepackage{epstopdf}
\usepackage{bm}
\usepackage{appendix}
\usepackage{multirow}
\usepackage{braket}
\usepackage[english]{babel}
\usepackage{hyperref}
\usepackage[capitalize]{cleveref}
\usepackage[square,numbers]{natbib}
\usepackage[T1]{fontenc}
\usepackage{xpatch}
\usepackage{tikz}
\usepackage{adjustbox}
\usepackage{xspace}
\usepackage[roman]{complexity}

\newcommand{\bvec}[1]{\mathbf{#1}}
\newcommand{\CC}{\mathbb{C}}

\newcommand{\vg}{\bvec{g}}

\newcommand{\vk}{\bvec{k}}

\newcommand{\vq}{\bvec{q}}
\newcommand{\vr}{\bvec{r}}

\newcommand{\vu}{\bvec{u}}

\newcommand{\vG}{\bvec{G}}

\newcommand{\vR}{\bvec{R}}

\newcommand{\vzero}{\boldsymbol{0}}

\newcommand{\diag}{\operatorname{diag}}

\renewcommand{\Re}{\operatorname{Re}}

\newcommand{\Tr}{\operatorname{Tr}}

\newcommand{\mc}[1]{\mathcal{#1}}

\newcommand{\ud}{\,\mathrm{d}}

\usepackage{xcolor}
\definecolor{purp}{RGB}{160, 32, 240}

\usetikzlibrary{fit}
\tikzset{%
  highlight/.style={rectangle,rounded corners,fill=blue!15,draw,fill opacity=0.3,thick,inner sep=0pt}
}

\global\long\def\R{\mathbb{R}}

\global\long\def\Tr{\mathrm{Tr}}

%% file: generalized_gs.bbl
\begin{thebibliography}{10}

\bibitem{BeckerEmbreeWittstenEtAl2021}
Simon Becker, Mark Embree, Jens Wittsten, and Maciej Zworski.
\newblock Spectral characterization of magic angles in twisted bilayer
  graphene.
\newblock {\em Phys. Rev. B}, 103(16):165113, 2021.

\bibitem{BeckerEmbreeWittstenEtAl2022}
Simon Becker, Mark Embree, Jens Wittsten, and Maciej Zworski.
\newblock Mathematics of magic angles in a model of twisted bilayer graphene.
\newblock {\em Prob. Math. Phys.}, 3(1):69--103, 2022.

\bibitem{BeckerHumbertZworski2022a}
Simon Becker, Tristan Humbert, and Maciej Zworski.
\newblock Fine structure of flat bands in a chiral model of magic angles.
\newblock {\em arXiv:2208.01628}, 2022.

\bibitem{BeckerLinStubbs2023}
Simon Becker, Lin Lin, and Kevin~D. Stubbs.
\newblock Exact ground state of interacting electrons in magic angle graphene,
  December 2023.

\bibitem{BernevigSongRegnaultEtAl2021}
B.~Andrei Bernevig, Zhi-Da Song, Nicolas Regnault, and Biao Lian.
\newblock {Twisted bilayer graphene. III. Interacting Hamiltonian and exact
  symmetries}.
\newblock {\em Phys. Rev. B}, 103(20):205413, 2021.

\bibitem{Bhatia1997}
Rajendra Bhatia.
\newblock {\em Matrix {{Analysis}}}, volume 169 of {\em Graduate {{Texts}} in
  {{Mathematics}}}.
\newblock {Springer New York}, {New York, NY}, 1997.

\bibitem{BistritzerMacDonald2011}
R.~Bistritzer and A.~H. MacDonald.
\newblock Moire bands in twisted double-layer graphene.
\newblock {\em Proceedings of the National Academy of Sciences},
  108(30):12233--12237, 2021.

\bibitem{BultinckKhalafLiuEtAl2020}
Nick Bultinck, Eslam Khalaf, Shang Liu, Shubhayu Chatterjee, Ashvin Vishwanath,
  and Michael~P. Zaletel.
\newblock Ground {{State}} and {{Hidden Symmetry}} of {{Magic-Angle Graphene}}
  at {{Even Integer Filling}}.
\newblock {\em Phys. Rev. X}, 10(3):031034, 2020.

\bibitem{CancesGarrigueGontier2023}
\'Eric Canc\`es, Louis Garrigue, and David Gontier.
\newblock Simple derivation of moir\'e-scale continuous models for twisted
  bilayer graphene.
\newblock {\em Phys. Rev. B}, 107:155403, Apr 2023.

\bibitem{2018Nature}
Yuan {Cao}, Valla {Fatemi}, Shiang {Fang}, Kenji {Watanabe}, Takashi
  {Taniguchi}, Efthimios {Kaxiras}, and Pablo {Jarillo-Herrero}.
\newblock {Unconventional superconductivity in magic-angle graphene
  superlattices}.
\newblock {\em Nature}, 556(7699):43--50, 2018.

\bibitem{ChatterjeeBultinckZaletel2020}
Shubhayu Chatterjee, Nick Bultinck, and Michael~P. Zaletel.
\newblock Symmetry breaking and skyrmionic transport in twisted bilayer
  graphene.
\newblock {\em Phys. Rev. B}, 101(16), apr 2020.

\bibitem{DasLuHerzog-Arbeitman2021}
Ipsita Das, Xiaobo Lu, Jonah Herzog-Arbeitman, Zhi-Da Song, Kenji Watanabe,
  Takashi Taniguchi, B.~Andrei Bernevig, and Dmitri~K. Efetov.
\newblock Symmetry-broken chern insulators and rashba-like landau-level
  crossings in magic-angle bilayer graphene.
\newblock {\em Nat. Phys.}, 17(6):710--714, mar 2021.

\bibitem{FaulstichStubbsZhuEtAl2023}
Fabian~M Faulstich, Kevin~D Stubbs, Qinyi Zhu, Tomohiro Soejima, Rohit Dilip,
  Huanchen Zhai, Raehyun Kim, Michael~P Zaletel, Garnet Kin-Lic Chan, and Lin
  Lin.
\newblock Interacting models for twisted bilayer graphene: A quantum chemistry
  approach.
\newblock {\em Phys. Rev. B}, 107(23):235123, 2023.

\bibitem{JiangLaiWatanabe2019}
Yuhang Jiang, Xinyuan Lai, Kenji Watanabe, Takashi Taniguchi, Kristjan Haule,
  Jinhai Mao, and Eva~Y. Andrei.
\newblock Charge order and broken rotational symmetry in magic-angle twisted
  bilayer graphene.
\newblock {\em Nature}, 573(7772):91--95, jul 2019.

\bibitem{LianSongRegnaultEtAl2021}
Biao Lian, Zhi-Da Song, Nicolas Regnault, Dmitri~K. Efetov, Ali Yazdani, and
  B.~Andrei Bernevig.
\newblock {Twisted bilayer graphene. IV. Exact insulator ground states and
  phase diagram}.
\newblock {\em Phys. Rev. B}, 103(20):205414, 2021.

\bibitem{LiuKhalafLee2021}
Shang Liu, Eslam Khalaf, Jong~Yeon Lee, and Ashvin Vishwanath.
\newblock Nematic topological semimetal and insulator in magic-angle bilayer
  graphene at charge neutrality.
\newblock {\em Phys. Rev. Res.}, 3(1), jan 2021.

\bibitem{LiuKhalafLeeEtAl2021}
Shang Liu, Eslam Khalaf, Jong~Yeon Lee, and Ashvin Vishwanath.
\newblock Nematic topological semimetal and insulator in magic-angle bilayer
  graphene at charge neutrality.
\newblock {\em Phys. Rev. Res.}, 3(1):013033, 2021.

\bibitem{nuckolls2023quantum}
Kevin~P Nuckolls, Ryan~L Lee, Myungchul Oh, Dillon Wong, Tomohiro Soejima,
  Jung~Pyo Hong, Dumitru C{\u{a}}lug{\u{a}}ru, Jonah Herzog-Arbeitman, B~Andrei
  Bernevig, Kenji Watanabe, et~al.
\newblock Quantum textures of the many-body wavefunctions in magic-angle
  graphene.
\newblock {\em Nature}, 620(7974):525--532, 2023.

\bibitem{PotaszXieMacDonald2021}
Pawel Potasz, Ming Xie, and Allan~H. MacDonald.
\newblock Exact {{Diagonalization}} for {{Magic-Angle Twisted Bilayer
  Graphene}}.
\newblock {\em Phys. Rev. Lett.}, 127(14):147203, 2021.

\bibitem{SaitoGeRademaker2021}
Yu~Saito, Jingyuan Ge, Louk Rademaker, Kenji Watanabe, Takashi Taniguchi,
  Dmitry~A. Abanin, and Andrea~F. Young.
\newblock Hofstadter subband ferromagnetism and symmetry-broken chern
  insulators in twisted bilayer graphene.
\newblock {\em Nat. Phys.}, 17(4):478--481, jan 2021.

\bibitem{NatSaito}
Yu~{Saito}, Jingyuan {Ge}, Kenji {Watanabe}, Takashi {Taniguchi}, and Andrea~F.
  {Young}.
\newblock {Independent superconductors and correlated insulators in twisted
  bilayer graphene}.
\newblock {\em Nature Physics}, 16(9):926--930, June 2020.

\bibitem{SoejimaParkerBultinckEtAl2020}
Tomohiro Soejima, Daniel~E. Parker, Nick Bultinck, Johannes Hauschild, and
  Michael~P. Zaletel.
\newblock Efficient simulation of moire materials using the density matrix
  renormalization group.
\newblock {\em Phys. Rev. B}, 102(20):205111, 2020.

\bibitem{TarnopolskyKruchkovVishwanath2019}
Grigory Tarnopolsky, Alex~Jura Kruchkov, and Ashvin Vishwanath.
\newblock Origin of magic angles in twisted bilayer graphene.
\newblock {\em Phys. Rev. Let.}, 122(10):106405, 2019.

\bibitem{tseng2022anomalous}
Chun-Chih Tseng, Xuetao Ma, Zhaoyu Liu, Kenji Watanabe, Takashi Taniguchi,
  Jiun-Haw Chu, and Matthew Yankowitz.
\newblock Anomalous hall effect at half filling in twisted bilayer graphene.
\newblock {\em Nature Physics}, 18(9):1038--1042, 2022.

\bibitem{wagner2022global}
Glenn Wagner, Yves~H Kwan, Nick Bultinck, Steven~H Simon, and SA~Parameswaran.
\newblock Global phase diagram of the normal state of twisted bilayer graphene.
\newblock {\em Physical Review Letters}, 128(15):156401, 2022.

\bibitem{WatsonKongMacDonaldEtAl2022}
Alexander~B. {Watson}, Tianyu {Kong}, Allan~H. {MacDonald}, and Mitchell
  {Luskin}.
\newblock {Bistritzer-MacDonald dynamics in twisted bilayer graphene}.
\newblock {\em Journal of Mathematical Physics}, 64(3):031502, March 2023.

\bibitem{WatsonLuskin2021}
Alexander~B Watson and Mitchell Luskin.
\newblock Existence of the first magic angle for the chiral model of bilayer
  graphene.
\newblock {\em J. Math. Phys.}, 62(9):091502, 2021.

\bibitem{WuSarma2020}
Fengcheng Wu and Sankar~Das Sarma.
\newblock Collective excitations of quantum anomalous hall ferromagnets in
  twisted bilayer graphene.
\newblock {\em Phys. Rev. Lett.}, 124(4), jan 2020.

\bibitem{xie2023phase}
Fang Xie, Jian Kang, B~Andrei Bernevig, Oskar Vafek, and Nicolas Regnault.
\newblock Phase diagram of twisted bilayer graphene at filling factor
  $\nu$=$\pm$3.
\newblock {\em Physical Review B}, 107(7):075156, 2023.

\bibitem{XieMacDonald2020}
Ming Xie and A.~H. MacDonald.
\newblock Nature of the {{Correlated Insulator States}} in {{Twisted Bilayer
  Graphene}}.
\newblock {\em Phys. Rev. Lett.}, 124(9):097601, 2020.

\bibitem{XieLianJackEtAl2019}
Yonglong Xie, Biao Lian, Berthold J{\"a}ck, Xiaomeng Liu, Cheng-Li Chiu, Kenji
  Watanabe, Takashi Taniguchi, B~Andrei Bernevig, and Ali Yazdani.
\newblock Spectroscopic signatures of many-body correlations in magic-angle
  twisted bilayer graphene.
\newblock {\em Nature}, 572(7767):101--105, 2019.

\bibitem{zhou2024kondo}
Geng-Dong Zhou, Yi-Jie Wang, Ninghua Tong, and Zhi-Da Song.
\newblock Kondo phase in twisted bilayer graphene.
\newblock {\em Physical Review B}, 109(4):045419, 2024.

\end{thebibliography}
